\documentclass[hidelinks,onefignum,onetabnum]{siamart220329}

\usepackage{amsmath,amssymb, graphicx, dsfont}
\usepackage{xcolor, soul}
\usepackage{tgpagella}
\usepackage[margin=1in]{geometry}
\usepackage[noend]{algpseudocode}
\usepackage{algorithm}
\usepackage{tikz}
\usetikzlibrary{calc}
\usepackage{hyperref}
\definecolor{blueish}{HTML}{007F99}
\definecolor{purpleish}{HTML}{72177A}
\hypersetup{
	colorlinks,
	citecolor=purpleish,
	filecolor=red,
	linkcolor=blueish,
	urlcolor=blue
}

\newtheorem{claim}{Claim}

\newtheorem{question}{Question}

\newcommand{\eps}{\epsilon}

\newcommand{\vol}{\operatorname{vol}}

\newcommand{\change}[1]{{#1}}
\renewcommand{\hl}[1]{{#1}}

\ifpdf
\hypersetup{
  pdftitle={Robust Factorizations and Colorings of Tensor Graphs},
  pdfauthor={J. Brakensiek and S. Davies}
}
\fi

\headers{Robust Factorizations and Colorings of Tensor Graphs}{J. Brakensiek and S. Davies}

\title{Robust Factorizations and Colorings of Tensor Graphs\thanks{Submitted to the editors February 2023.}}

\author{Joshua Brakensiek\thanks{Stanford University.
  (\email{jbrakens@cs.stanford.edu}). Funded in part by an NSF Graduate
  Research Fellowship and a Microsoft Research PhD Fellowship.}
\and Sami Davies\thanks{University of California, Berkeley. (\email{samidavies@berkeley.edu}). Funded in part by a Microsoft Research PhD Fellowship and an NSF CI Fellowship. Part of the work was done while this author was at Northwestern University.}}

\usepackage{amsopn}

\begin{document}

\maketitle

\begin{abstract}
Since the seminal result of Karger, Motwani, and Sudan,
algorithms for approximate 3-coloring have primarily centered around SDP-based rounding.
However, it is likely that important combinatorial or algebraic insights are 
needed in order to break the $n^{o(1)}$ threshold. 
One way to develop new understanding in graph coloring is to study 
special subclasses of graphs. 
For instance, Blum studied the 3-coloring of random graphs, 
and Arora and Ge studied the 3-coloring of graphs with low threshold-rank.

In this work, we study graphs which arise from a tensor product, which
appear to be novel instances of the 3-coloring problem.  We consider
graphs of the form $H = (V,E)$ with $V =V( K_3 \times G)$ and
$E = E(K_3 \times G) \setminus E'$, where
$E' \subseteq E(K_3 \times G)$ is any edge set such that no vertex has
more than an $\epsilon$ fraction of its edges in $E'$.  We show that
one can construct $\widetilde{H} = K_3 \times \widetilde{G}$ with
$V(\widetilde{H}) = V(H)$ that is close to $H$.  For arbitrary $G$,
$\widetilde{H}$ satisfies
$|E(H) \Delta E(\widetilde{H})| \leq O(\epsilon|E(H)|)$.  Additionally
when $G$ is a mild expander, we provide a 3-coloring for $H$ in
polynomial time. These results partially generalize an exact tensor
factorization algorithm of
Imrich.
On the other hand, without any assumptions on $G$, we show that it is
\textsf{NP}-hard to 3-color $H$.
\end{abstract}



\section{Introduction}

The 3-coloring problem is one of the most classical problems in
theoretical computer science \cite{karp1975computational}. Although it
is NP-hard, much effort has been made to understand the
\emph{approximate} 3-coloring problem: given a 3-colorable
graph\footnote{In this paper, all graphs are undirected, simple (i.e., no double
  edges) and loopless (no edge from a vertex to itself).} on
$n$ vertices as input, what is the fewest number of colors one can
efficiently color the graph with?  Initially, combinatorial algorithms
were dominant in approximate 3-coloring, bringing us Wigderson's
famous $O(\sqrt{n})$-approximation \cite{wigderson1983improving}, as
well as Blum's $\widetilde{O}(n^{3/8})$-approximation and Blum's
3-coloring algorithm for many random 3-colorable graphs \cite{Blum94}.
Then, SDP-algorithms took center stage\footnote{The one exception to
  this trend is Kawarabayashi and Thorup's
  $\widetilde{O}(n^{4/11})$-approximation
  \cite{kawarabayashi2012combinatorial}.}.  This began with the
celebrated work of Karger, Motwani, and Sudan
\cite{karger1998approximate}.  With a lot of extra work, clever
observations were made to augment their algorithm or combine it with
combinatorial algorithms and obtain better approximation
results \cite{AC06,Chlamtac07, kawarabayashi2017coloring}.  We give a
more complete history of the 3-coloring problem in Section \ref{sec:
  related-work}.

It is unclear how much further success can be obtained by combining 
SDP algorithms with fancier combinatorial techniques, 
and it is likely that completely new ideas are needed.
One way to continue building new insight on 3-coloring is to focus on interesting subclasses of 3-colorable graphs. 
Indeed, Blum did just this in targeting random 3-colorable graphs \cite{Blum94}, 
and Arora and Ge generalized this result by studying low threshold rank 3-colorable graphs \cite{arora2011new}. 
Additionally, the improvement from Kawarabayashi and Thorup comes 
from focusing on graphs with high degree \cite{kawarabayashi2017coloring}.
Overall, we do not fully understand what properties make graphs easy or hard to color, 
and our work is a further exploration of this.

In this work, we propose a new class of 3-colorable graphs that are
interesting for the 3-coloring problem: graphs that are close to the
tensor graph $K_3 \times G$.  For undirected graphs $F$ and $G$, their
\emph{tensor product}\footnote{Also known as the cardinal, direct, or
  Kronecker product, among other names.} $F \times G$ is a graph on
the vertex set $V(F) \times V(G)$, where vertices $(f,g)$ and
$(f',g')$ are incident if and only if $(f,f') \in E(F)$ and
$(g,g') \in E(G)$. Observe that if $G$
is connected, we can set $F=K_3$ and the graph $P = K_3 \times G$ is
easy to 3-color. In particular, we first locate (say by brute force)
one triple of the form $K_3 \times \{g\}$ for some $g \in G$ in $P$,
which we call a \emph{core} triple. We color the core triple with
three distinct colors, and then observe the colors of the neighbors in
the graph are forced. That is, for any $g'$ in $g$'s
neighborhood, the core triples $K_3 \times \{g\}$ and
$K_3 \times \{g'\}$ have 6 edges between them such that if $g$'s copy
of $K_3$ is colored with three distinct colors, there is only one
valid coloring for $g'$'s copy.  This coloring propagates through out
$P$.  See Figure \ref{fig: coloring} for an illustration.  On the
other hand, suppose we delete edges from $P$ to form a graph $H$.  If
the number of deletions is large enough that a coloring is not
immediately forced from fixing the colors on one core triple, then
it is not obvious how to 3-color $H$. 

Another way to view this is via LP hierarchies.  In particular,
consider a $O(1)$-level Sherali-Adams lift on the basic 3-coloring LP.
This simple coloring algorithm that was successful for 3-coloring $P$
is also successful for 3-coloring $H$ exactly when the lifted LP
variables directly provide the core triples.  In fact, we prove this
algorithm is successful when $G$ is an expander for our deletion model
in Theorem \ref{thm:no-sparse-cuts}.  Instead of arguing that a lifted
LP/ SDP rounding procedure could succeed through properties of the LP
solution\change{--}we do not know whether such a proof is possible in our
setting\change{--}we study its combinatorial (or topological!) analog.  We
believe this is a more promising avenue to complement the existing
3-coloring work, and overall lead to more progress on the problem.

\begin{figure}
  \begin{center}
    \includegraphics[width = 11cm]{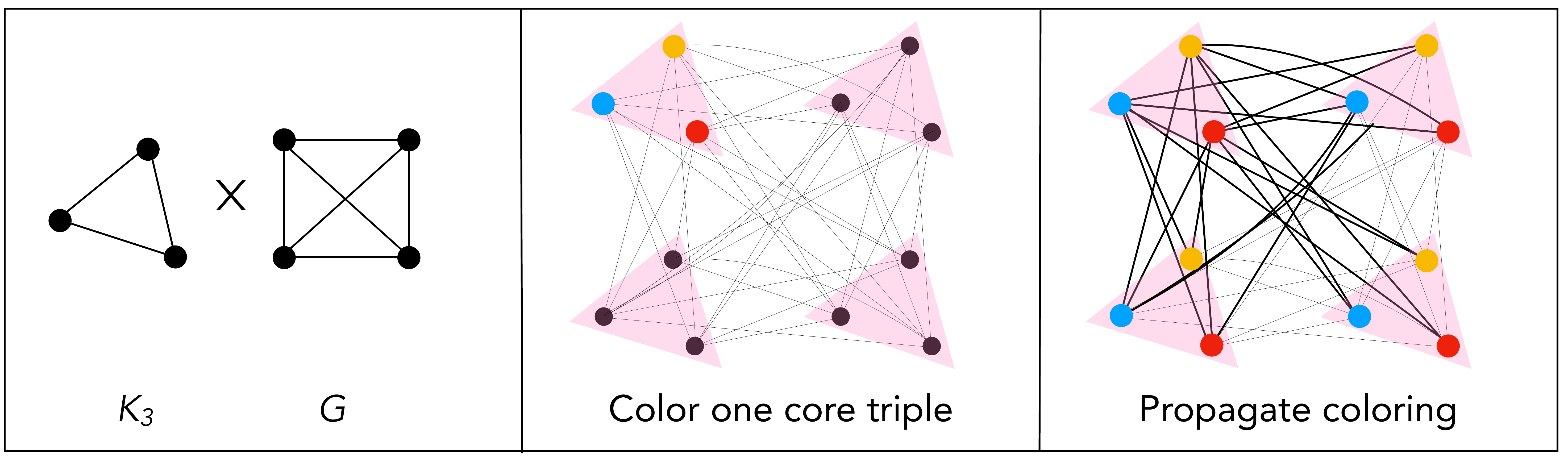}
  \end{center}
  
  \caption{To color $K_3 \times G$, fix a \change{set of 3 vertices} with 3 colors. 
  Then, see what colors this forces on the rest of the graph. Some initial \change{choice of 3 vertices} will induce a valid 3-coloring.}
  \label{fig: coloring}
  \end{figure}

However, this seemingly simple family of graphs does not seem
to be properly captured by the current literature on coloring algorithms.
In particular, to the best of our knowledge,
there is no guarantee that previous analyses of coloring algorithms would work well on $H$.
First, $H$ can look far from random, which prohibits a
guaranteed success by Blum's coloring tools~\cite{Blum94}. 
Additionally, the threshold rank of $H$ is uncontrollable. 
If $G$ has high threshold rank, 
then the tensor $K_3 \times G$ has high threshold rank, 
and the deletion of edges to form $H$ has varied effects on the spectrum.
Overall, this means we have no guarantee for
a polynomial time $O(\log n)$-coloring by the algorithm from Arora and
Ge \cite{arora2011new}\change{.}\footnote{It is possible that there 
exists some reweighting of the edges of $H$ such that the reweighted graph has at most polylog$(n)$ threshold rank. 
This leaves open the possibility that a combination of combinatorial techniques
together with the algorithm by Arora and Ge \cite{arora2011new} could produce a quasi-polynomial time algorithm using $O(\log n)$ many colors.}
To the best of our knowledge, current analyses of Gaussian cap rounding procedures 
would not color near tensors in constant colors.
In response, we ask the following question:

\bigskip
{\centering 
\emph{Suppose edges are deleted from $P = K_3 \times G$ to build $H$. 
When can one 3-color $H$ in polynomial time?}\par
}
\bigskip

If no edges are deleted from the tensor product, polynomial time
factorization is possible.  We will describe an algorithm presented by
Imrich \cite{imrich1998factoring} in Section~\ref{Imrich-details},
which shows that if $F$ and $G$ are connected non-bipartite, and prime
with respect to the tensor product \footnote{ I.e. there is no such
  graphs $F_1,F_2$ on more than 1 vertex where $F=F_1 \times F_2$, and
  same for $G$.} one can reduce to a factoring problem for the
\emph{Cartesian product}.\footnote{The Cartesian product of two graphs
  $F, G$ is a graph $F \square G$ on $V(F) \times V(G)$ such that
  $(u_1,v_1) \sim (u_2,v_2)$ in $F \square G$ if $(u_1,u_2) \in F$ and
  $v_1 = v_2$ or $u_1 = u_2$ and $(v_1,v_2) \in F$. Note that unless
  $F$ and $G$ contain loops, $F \times G$ and $F \square G$ have
  disjoint sets of edges.} In particular, Imrich constructs a
Cartesian product graph $F' \square G'$ (without knowing what $F'$ and
$G'$ are), where $V(F') = V(F)$ and $V(G')=V(G)$.  Then, a similarity
metric by Winkler~\cite{winkler1987factoring} (see also
\cite{hochstrasser1992note,imrich2007recognizing, FeigenbaumHS85}) can
be used on $F' \square G' $ to identify components, one of which will
build the graph $F$ and the other $G$.  Until then, one important
thing to note about these procedures is that they are very brittle to
any deviations from a tensor graph.  In other words, these procedures
cannot be run on graphs $H$ which are very close to a tensor, e.g.
adding a small number of edges to $H$ would turn it into a tensor.
Tensor product graphs (and graphs close to a tensor product) are in a
wide range of applications, including image
processing~\cite{signal-processing}, network design~\cite{LCKFG10},
complex datasets~\cite{big-data}, dynamic location
theory~\cite{hammack2011handbook}, and chemical graph
theory~\cite{wienertrees}.  We note that graphs near tensor products
are especially important in modeling social networks \cite{big-data}.
Therefore, an interesting combinatorial question is

\bigskip
{\centering
\emph{When can one approximately factor a graph that is close to a tensor product in polynomial time?}\par
}
\bigskip

We now present some models and results on these questions.

\subsection{Problem and theorem statements}

Due to both the motivation from 3-coloring and our interest in robust
tensor graph factoring algorithms, we study the following question.
Let $K_3$ be the clique on three vertices, $G$ be a graph on $n$
vertices, and $P = K_3 \times G$ be their tensor product.  We consider
graphs of the form $H = (V(P), E(P)\setminus E')$, for
$E' \subseteq E(P)$ an edge set such that no vertex $v$ in $H$ has
more than an $\epsilon$ fraction of its incident edges in $E'$.  We
say that such an $H$ is \emph{$\epsilon$-near} the triangle tensor
product $P$.  Our deletion model is very general, as the deleted edges
could be \change{adversarially} chosen in such a way that nodes with
substantially different looking neighborhoods in $P$ look the same in
$H$.  
This model appears to be novel, as approximate graph products have mainly only been studied for the
Cartesian product, not for the tensor product--we will discuss the
related work more in Subsection \ref{sec: related-work}.

Suppose we are given $H$ as above. Our primary reconstruction goal is
the \emph{$\ell_1$ reconstruction goal},\footnote{As the name suggests, there are other natural
reconstruction goals, see Section~\ref{sec:conclusion} for more
details.} in which we aim to construct a graph $\widetilde{H} = K_3 \times \widetilde{G}$
with $V(\widetilde{H}) = V(H)$ and
\[|E(H) \Delta E(\widetilde{H})| \leq O(\epsilon|E(H)|).\] 

We prove the following algorithmic results in Section \ref{sec: proof-alg}.
Our first theorem holds for any $G$.

\begin{theorem}
\label{thm:main}
Assume $\eps = \Omega(|V(H)|/|E(H)|)$. 
Let $H$ be $\epsilon$-near $K_3 \times G$.  Then, there is an
algorithm running in time $O(n^6)$ that constructs a tensor
$\widetilde{H} \cong K_3 \times \widetilde{G}$ with $V(\widetilde{H})
= V(H)$ achieving the $\ell_1$ reconstruction goal.
\end{theorem}

\hl{Our second theorem holds when $G$ is an expander. }For any vertex $v
\in V(G)$, let $\Gamma_G(v) = \{u : (u,v) \in G\}$ denote the neighbors of $v$.
We will call $G$ an \emph{$\alpha$-edge-expander} if for every $S \subseteq
G$, the following holds
\[
|E(G) \cap (S \times \bar{S})| \ge \alpha \min\left(\sum_{v \in S}
  |\Gamma_G(v)|, \sum_{v \in \bar{S}} |\Gamma_G(v)|\right).
\]

\begin{theorem}
  \label{thm:no-sparse-cuts}
  Fix $\eps < 1/40$.
  Let $H$ be $\epsilon$-near $K_3 \times G$, where graph $G$ is a $3\epsilon$-edge-expander.
  Then, there is an algorithm running in time $\textsf{poly}(n)$ that
  extracts a valid 3-coloring of $H$.
\end{theorem}

As we shall see, this theorem is possible as the expansion of $G$
allows for \change{the} partial three-coloring found in Theorem~\ref{thm:main} to
be converted  into a globally consistent three-coloring. We can replace
the condition in Theorem~\ref{thm:no-sparse-cuts} that $G$ is a
$3\epsilon$-edge-expander with a small-set expansion condition on
$G$. See Section~\ref{sec:ext} for more details.

We complement our algorithmic results by showing that 3-coloring $H$ is hard for general $G$.
\begin{theorem}\label{thm:no-3-col}
  Given as input \change{a} graph $\widetilde{H}$ which is $\eps$-near
  $K_3 \times G$ for some $G$, it is \textsf{NP}-hard to find a
  $3$-coloring of $\widetilde{H}$.
\end{theorem}

Our algorithms are combinatorial, although they seem to draw on some
topological properties of graph tensors (see the technical overview).

\subsection{Related work}\label{sec: related-work}

As our work bridges the rich theory of graph factoring algorithms
with the world of approximate graph coloring, there are a number of
prior works which relate to our investigation.

\paragraph{Approximate coloring algorithms} We refer the reader to the introduction in \cite{kawarabayashi2017coloring} 
for a detailed recap of 3-coloring progress over the past several decades.
A simple algorithm by Wigderson \cite{wigderson1983improving} uses $O(\sqrt{n})$ colors
and has remained an important subroutine in coloring graphs with high degree \cite{karger1998approximate}.
\change{The barrier of $O(\sqrt{n})$ colors was broken by Berger and Rompel \cite{berger1990better}.
Blum introduced some intricate combinatorial techniques
which inspired much future coloring work, }
in particular those of Kawarabayashi and Thorup \cite{kawarabayashi2012combinatorial,kawarabayashi2017coloring}.
Since the seminal result of Karger, Motwani, and Sudan \cite{karger1998approximate}, 
algorithmic results in approximate graph coloring have focused on SDP based
algorithms \cite{AC06, Chlamtac07,arora2011new,kawarabayashi2017coloring}.
Several integrality gap instances for the original 3-coloring SDP formulation
presented by Karger, Motwani, and Sudan~\cite{karger1998approximate}
have been found by Karger, Motwani, and Sudan
\cite{karger1998approximate}; Frankl et. al~\cite{frankl-intgap,
arora2011new} ; and Feige et. al~\cite{feige-intgap}.  These graphs
have valid solutions to the SDP, but have chromatic number at least
$n^{c}$, for small constant $c>.01$ and $n$ the size of the vertex
set.  On the other hand, if a 3-colorable graph has \emph{threshold rank} $D$, 
i.e. once its eigenvalues are scaled to be in $[-1,1]$,
at most $D$ of them are less than $-1/16$, then one can color the graph in
$O(\log n)$ colors in time $n^{O(D)}$ \cite{arora2011new}.

\paragraph{Hardness of approximate graph coloring} The hardness results for approximate graph coloring are quite far from
the algorithmic results. For a 3-colorable graph it is known that it
is NP-hard to color with 5 colors~\cite{BBKO21}, beating a
long-standing record of 4 colors~\cite{KLS93,GK04,BG16}. For
$k$-colorable graphs, it is NP-hard to color with
$\binom{k}{\lfloor k/2\rfloor}$ colors~\cite{WZ20}. Under a variety of
conditional assumptions, it is hard to color a 3-colorable graph with
superconstant colors (up to roughly $\operatorname{polylog} n$),
see~\cite{DMR09,DS10,WZ20,GS20}. 

\paragraph{Graph factoring algorithms} Given a graph product, one can
efficiently factor a graph with respect to the product into prime
graphs, i.e. graphs that cannot be further non-trivially factored
according to the product.  For the Cartesian product, any finite,
simple, connected graph can be factored in polynomial time
~\cite{FeigenbaumHS85, winkler1987factoring, Feder92,
  imrich2007recognizing}.  Moreover, Imrich and
Peterin~\cite{imrich2007recognizing} present an algorithm that
performs this factorization in \hl{time linear in the number of}
edges.  Similarly for the strong product, whose graph unions the edges
of the Cartesian product and the tensor product graphs, a polynomial
time factorization was found by Feigenbaum and Sch{\"{a}}ffer
\cite{FeigenbaumS92-strong} for finite, simple, connected graphs.  To
factor a graph with respect to the tensor product, one can combine
algorithms of Imrich and Winkler \cite{imrich1998factoring,
  winkler1987factoring}; we detail these algorithms--and how they can
be combined to decompose a tensor product--at the end of Section
\ref{Imrich-details}.
For both the Cartesian product and the strong product, prime
factorizations (i.e., a decomposition into irreducible factors) are
unique for finite, connected, simple graphs
~\cite{hammack2011handbook}.  In the case of the tensor product, we
are required to make the additional assumption that the graph is
non-bipartite, as without it, a unique factorization with respect to the tensor product does not exist.
With this additional assumption, the unique
factorization can also be found in polynomial time
\cite{imrich1998factoring}. 

\paragraph{Approximate graph products} The theory of approximate graph products has
also been studied before, but in different settings and with different
goals from us.  For a graph $H$ that is close to a product graph $H'$,
i.e., it takes a small number of edge deletions or insertions to
transform $H$ to the product graph $H'$, previous works study how to
find $H'$ \cite{ZmazekZ01,ZmazekZ07-strong, HellmuthIKS09,
  HIKS09-strong, HellmuthIK13-star}.  Feigenbaum and
Haddad~\cite{FeigenbaumH89-factorable} showed that for the Cartesian
product, obtaining such an $H'$ with the fewest possible edge
insertions or the fewest possible edge deletions is \textsf{NP}-hard.
Overall, the Cartesian product is the most well studied graph product,
and both approximate Cartesian product and strong product graphs have
connections to theoretical biology, as they model evolutionary
relationships of observable characteristics
\cite{HellmuthIKS09,HIKS09-strong}.  To the best of our knowledge,
prior to our work, approximate graph products have only been studied
for the Cartesian product and the strong product, and not with respect
to the tensor product.  The closest problem is the Nearest Kronecker
Product problem, which given $A \in \mathbb{R}^{m \times n}$ seeks to
find $B \in \mathbb{R}^{m_1 \times n_1}$ and
$C \in \mathbb{R}^{m_2 \times n_2}$ such that $||A - B \times C||_F$
is minimized, for $B \times C $ the Kronecker product of $B$ and $C$
and $m = m_1 \cdot m_2$, $n = n_1 \cdot n_2$
\cite{Kronecker-approx,Kronecker-product}.  Note that the relation
between this problem and the approximate tensor product problem lies
in the fact that the adjacency matrix of the tensor product of two
graphs is the Kronecker product of the underlying adjacency matrices.
Another similar problem is the closest separable state problem in
Quantum systems, which approximates the entanglement of a system by
measuring is how far it is from a composite of separable states
\cite{modi2010unified, wiesniak2020distance}.
More on product graphs, their factorizations, and approximate graph
products can be found in the book by Hammack, Imrich, and
Klav{\v{z}}ar \cite{hammack2011handbook}.

\paragraph{Learning theory} A related line of work to ours is tensor
decomposition in the learning theory community.  Tensors (not
necessarily tensor graphs, just tensors) represent higher order
information from data.  A common goal is uncover the latent (hidden)
variables underlying some data in order to understand it in a lower
dimensional form \cite{moitra2014algorithmic}.  One popular way to
achieve this is the CP (CANDECOMP/ PARAFAC) decomposition, which
writes a tensor as a sum of rank 1 tensors
\cite{roughgarden2021beyond, janzamin2019spectral}.  Other related
decompositions are PCA, Tensor Robust PCA, and the Tucker
decomposition \cite{roughgarden2021beyond, janzamin2019spectral,
  lu2016tensor}. A similar research topic is reconstructing a
partially observed tensor (e.g.,~\cite{zhang2018recovery}).

\subsection{Technical overview}

We now present overviews for the proofs of Theorems \ref{thm:main}, 
\ref{thm:no-sparse-cuts}, and \ref{thm:no-3-col}.
The full proofs for Theorems \ref{thm:main} and \ref{thm:no-sparse-cuts} are in Section \ref{sec: proof-alg}, 
and the full proof for Theorem \ref{thm:no-3-col} is in Section
\ref{sec: hardness}.

\subsubsection{Overview of reconstructing a tensor graph}\label{Imrich-details}

To give intuition for our results, we first summarize Imrich's
algorithm for efficiently factoring a tensor $P = F \times G$. The
first goal is to find a graph $S$ on the same vertex set as $P$ that
is isomorphic to a Cartesian product $S = F' \square G'$. The procedure
Imrich uses to construct $S$ was inspired by an algorithm of
Feigenbaum and Sch\"affer \cite{FeigenbaumS92-strong} for the strong
product.

The key to constructing $S$ is the following observation on
intersections of neighborhoods in tensor graphs.  For any vertices
$u, v$ in $P$, let $I_P(u,v) = \Gamma_P(u) \cap \Gamma_P(v)$. Since
$P$ is a tensor graph, one can show that
$I_P(u,v) = I_F(u_f,v_f) \times I_G(u_g,v_g)$, where $u_f$ is the
projection of $u$ onto $V(F)$, etc. In particular, if $I_P(u,v)$ is
maximal (as a set) among $v \neq u$, then it must be that either
$u_f = v_f$ or $u_g = v_g$. In particular, adding all such maximal
edges $(u,v)$ to $S$ will keep $S$ consistent with a Cartesian product
$F' \square G'$, where $F'$ is on the same vertices as $F$ and $G'$ is
on the same vertices as $G$. If $S$ is not a connected graph, we
repeat this procedure, where the maximal $I_P(u,v)$'s are found for
$v$'s which are not in the same connected component as
$u$.\footnote{There are other considerations which Imrich carefully
  handles, such as if there is a third vertex $w$ with
  $I_P(u,v) = I_P(u,w)$ and $\Gamma_P(w) \subseteq \Gamma_P(v)$.}

\begin{figure}[h]
  \begin{center}
\includegraphics[width = 11cm]{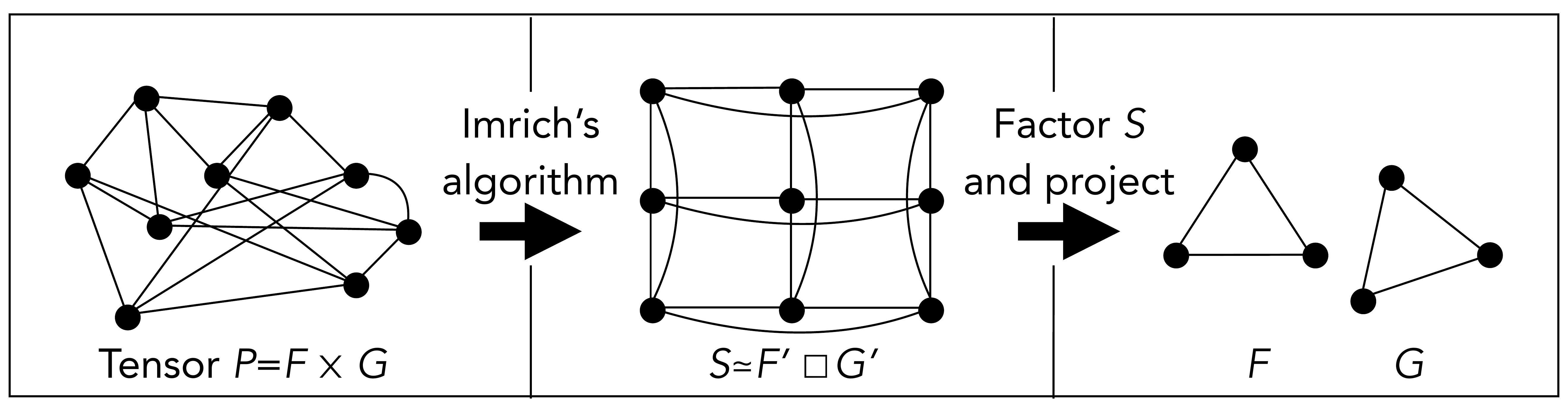}
  \end{center}
  \caption{An illustration of Imrich's algorithm. 
  The algorithm is given an unknown connected, nonbipartite tensor $P$. 
Imrich's algorithm builds $S\cong F'  \square G'$, with $V(G) = V(G')$ and $V(F) = V(F')$. 
Then, the algorithm factors $S$ to find $F'$ and $G'$, and projects $P$ onto $F'$ and $G'$ to find $F$ and $G$.}
  \label{fig: Imrich}
\end{figure}

Once $S$ is constructed, factoring the Cartesian product can be done
with a variety of algorithms, including the one of
Winkler~\cite{winkler1987factoring}. Winkler's algorithm is
rather elegant. Let $d_S(u,v)$, the length of the shortest path
between $u$ and $v$ in $S$. Define $(u_1, v_1), (u_2, v_2) \in S$ to
be \emph{similar} if $d_S(u_1, u_2) + d_S(v_1, v_2) \neq d_S(v_1, u_2) + d_S(u_1, v_2).$
While this similarity relation may not be transitive, we can perform
a depth first search to find all components that are
\emph{transitively similar}. The components found from Winkler's algorithm either correspond precisely
to the Cartesian product factors $F'$ and $G'$
of $S$, with $V(F') = V(F)$ and $V(G') = V(G)$,
or $S$ (and thus $P$) can be decomposed into
more than two factors. From this
factorization, we can extract $F$ and $G$ from $P$ with the following
projection trick. Given $u_f,v_f \in V(F')$ and $u_g,v_g \in V(G')$
with $((u_f,u_g),(v_f,v_g)) \in E(P)$, we add $(u_f,v_f)$ to $E(F)$ and
$(u_g,v_g)$ to $E(G)$. See Figure \ref{fig: Imrich} for an illustration. 
This completes the factoring algorithm.

\subsubsection{Reconstruction Algorithm}

To approximately factor a graph $H$, we need to approximate for each
vertex of $v$ the $t \in K_3$ and $g \in G$ corresponding to $v$,
which we call the ``color class'' and ``$G$ class'' of $v$
respectively. Since $G$ is unknown, we first try to group the vertices
of $H$ into triangles, which estimate $K_3 \times \{g\}$ for some $g \in G$, 
then use the edges between these triangles to estimate the edges of $G$.

\paragraph{Candidate edge graph} For Imrich, it sufficed to connect
pairs of nodes in the surrogate Cartesian product graph whose
neighborhoods' intersection in $P$ satisfied some maximality
criteria. To make this maximality criteria more robust, we define what
is known as the \emph{candidate edge graph $C$} on $V(H)$. Informally,
two vertices $u, v \in V(H)$ form an edge of $C$ if the intersection
of their neighborhoods has size approximately half their degrees (see
Section~\ref{sec: graphs}). Note that the edges of $C$ and the edges
of $H$ are qualitatively quite different (they can even be
disjoint). However, a key property of $C$ is that for every vertex
$g \in G$, the three vertices of $H$ corresponding to
$K_3 \times \{g\}$ form a triangle in $C$. We call this a \emph{core}
triangle.

We show the triangles of $C$, which we call $\mathcal T(C)$, have a
very particular form.  Such a triple is one of two types: it is either
(1) ``close'' to some core triangle (which we call \emph{quasi-core})
or (2) contains vertices whose color classes are all the same and
whose $G$ classes have very structured pairwise intersection (which we
call \emph{monochrome}, see Lemma \ref{lem:tri-types}). Because the
color classes of $H$ are hidden information, we cannot directly
determine if any triangle of $C$ is quasi-core or monochrome.

\paragraph{Triangle components} Instead, we separate these two types
of triangles topologically. We say that two triangles of
$\mathcal T(C)$ are \emph{compatible} if the subgraph of $H$ on the
six vertices of these triangles is isomorphic to $K_3 \times K_2$. In
other words, it is consistent that these two triangles correspond to
$K_3 \times \{g\}$ and $K_3 \times \{g'\}$ for some
$(g, g') \in E(H)$. This compatible relation divides $\mathcal{T}(C)$
into connected components; specifically, one can build a graph with vertex set $\mathcal{T}(C)$, 
and where two triangles in $\mathcal{T}(C)$ are connected exactly when they are compatible. 
Perhaps the most crucial (although easy to
prove) technical lemma of this paper is that each component of
$\mathcal T(C)$ consists only of quasi-core triangles (which we call a
\emph{core component}) or only of monochrome triangles (see
Lemma~\ref{lem:components}).

\paragraph{Coloring algorithm} Assume we pick an arbitrary component
$Y_j$ of $\mathcal T(C)$. Let $U_j$ be the vertices of $H$ covered by
$Y_j$. By guessing the colors of one of the triangles of $Y_j$ and
then performing a depth-first search, we can efficiently color all the
vertices of $H[U_j]$ (i.e., the subgraph of $H$ induced by $U_j$); or,
if it fails, we can deduce that $Y_j$ is not a core component (see
Proposition~\ref{prop:core-color}). If $G$ is a sufficiently good
expander, then we can show that $H$ has enough edges to force there to
be only a single core component (although there can be a large number
of monochrome components). Thus, by looping over all possible $Y_j$'s
(of which their are clearly at most $O(n^3)$), our coloring algorithm
will succeed on one of them, proving Theorem~\ref{thm:no-sparse-cuts}.

\paragraph{Matching algorithm to build triangles} By finding the
valid 3-coloring of $H[U_j]$, we can augment this by finding an
approximate tensor factorization of $H[U_j]$. This is the heart of
Lemma~\ref{lem:quasi-core-match}. The key idea is we take the three
color classes of $H[U_j]$, which we call
$\mathcal A, \mathcal B, \mathcal C$, and perform a tripartite
matching algorithm on them.  More formally, one can build a weighted
tripartite graph on
$(\mathcal{A}\dot{\cup }\mathcal{B}\dot{\cup } \mathcal{C}, E)$, where
the bipartite subgraphs between $\mathcal{A}$ and $\mathcal{B}$, as
well as $\mathcal{B}$ and $\mathcal{C}$, are complete, with each edge
having weight corresponding to their pairwise intersection, normalized
by the degrees of the vertices. By finding a max-min matching,\footnote{\change{Also known as a bottleneck matching.}} that is
a matching which maximizes the weight of the minimum weight edge
(which can be done in polynomial time), on this tripartite graph, we
find a chosen set of triples on the vertices in $Y_j$. In particular,
assuming $Y_j$ is a core component, each of the triples found will be
approximately quasi-core triangles. We can then build a tensor graph
corresponding to $H[U_j]$ by having each triple found by the matching
correspond to a vertex of the reconstructed graph $\widetilde{G}$, and
have each edge of $\widetilde{G}$ correspond to any pair of triples
sharing at least one edge in $H$. Analyzing this factorization
accurately requires showing that the every error in the reconstruction
can be ``charged'' to a mismatch in the neighborhoods of the matched
triples. The max-min guarantee implies that the number of such
mismatches is of $O(\eps)$ edge density, allowing us to achieve the
$\ell_1$ reconstruction goal for this subgraph.

\paragraph{Finishing the factorization} To factor the whole graph, we recursively apply
Lemma~\ref{lem:quasi-core-match}. In particular, we loop through the
triangle components of $\mathcal T(C)$. If the lemma successfully
factors the subgraph of $H$ induced by that component \emph{and} the
cut between that subgraph and the rest of $H$ has sufficiently few
edges, we recurse on the remainder of $H$.  Our tensor graph is then
the disjoint union of the tensors products found for each subgraph. By
combining the reconstruction guarantees for each connected component,
we have the factorization achieves the $\ell_1$ reconstruction goal,
proving Theorem \ref{thm:main}.

Note that we cannot easily achieve a 3-coloring through such a
recursive algorithm, as even though each component is correctly
3-colored, the sparse edges between the components make finding a
globally consistent coloring intractable. We formalize this hardness
in Theorem~\ref{thm:no-3-col}.

\subsubsection{Hardness}

The proof of Theorem~\ref{thm:no-3-col} is ultimately a reduction from
3-coloring, albeit in a roundabout manner. Given an instance $G$ of
the 3-coloring problem, we replace each vertex of $u \in G$ with
\change{27 copies $(u, x)$ for $x \in [3]^3$ corresponding to a copy
  of $K_3 \times K_3 \times K_3$.}  By a result of Greenwell and
Lovasz~\cite{Greenwell1974}, there are three types of 3-colorings of
\change{$K_3 \times K_3 \times K_3$}, arising from the three different
triangles in the tensor product. For each edge
\change{$\{u,v\} \in E(G)$ in the base graph}, we add edges between
their \change{copies} such that any valid $3$-colorings of
\change{these copies} must have different types. One can show that $G$ is
3-colorable if and only if the graph produced by the gadget reduction
is. Many similar gadget reductions have been performed in the hardness
of approximation literature, such as in \cite{DMR09}.

One can show that if $G$ is 3-colorable, then the graph resulting from
this reduction is a subset of $K_3 \times G'$ from some graph
$G'$. However, the reduction may not be $\eps$-near to this tensor.
To circumvent this, \change{instead of reducing from 3-coloring directly, 
we reduce from ``3-coloring with equality''\footnote{\change{This name was coined by a reviewer.}}, 
where each vertex now has many copies which are forced to be equal by equality constraints. 
There are so many copies that less than an $\eps$ fraction of the edges from each vertex correspond to 3-coloring constraints. 
If we repeat the aforementioned reduction, we then obtain a graph which is $\eps$-near to a tensor of the form $K_3 \times G'$.}

One can show that if $G$ is 3-colorable, then this graph is
$\eps$-near a $K_3$ tensor. However, if $G$ is not 3-colorable, then
the resulting graph isn't even \change{3-}colorable. This is enough to
establish Theorem~\ref{thm:no-3-col}, for if we had a polynomial-time
algorithm for 3-coloring these graphs $\eps$-near a $K_3$ tensor, then
we could solve the general 3-coloring problem.

\subsection{Paper outline}

In Section~\ref{sec: prelims}, we define the necessary terms and
concepts needed to prove our algorithmic results. Then, in Section
\ref{sec: proof-alg} we prove the main algorithmic results:
Theorem~\ref{thm:main} and Theorem~\ref{thm:no-sparse-cuts}. Section
\ref{sec: hardness} contains the proof of our main hardness theorem:
Theorem \ref{thm:no-3-col}. We conclude in
Section~\ref{sec:conclusion} with directions for future work\hl{.}

\section{Algorithm Preliminaries}\label{sec: prelims}

In this section, we give the key definitions and other concepts needed
to prove the algorithmic results.

\subsection{Important definitions and propositions}\label{sec: defs}

We let $P= K_3 \times G$ denote the tensor product.
As in the problem statement, the graph $H$ is $\epsilon$-near a triangle tensor, 
as the following definition details.

\begin{definition}
  \change{A g}raph $H$ is \emph{$\epsilon$-near a triangle tensor} if 
  there exists some product $P$ with
  $H \cong (V(P),E(P) \setminus E')$, 
  for $E'$ a set of edges where no vertex in $P$ has more than
  an $\epsilon$ fraction of its incident edges in $E'$.
\end{definition}
Let $X$ be an arbitrary graph.
We denote the neighborhood of a vertex $v$ of $X$ by $\Gamma_X(v)$,
where we specify the graph in which we take the neighborhood in the
subscript. 
Additionally, for every subset $S \subseteq V(X)$, define 
\[
  \vol_X(S) = \sum_{v \in S} |\Gamma_X(v)|.
\]

Recall that Theorem \ref{thm:no-sparse-cuts} considers
graphs whose tensor component $G$ is an $\alpha$-edge-expander. That
is, for every $S \subseteq V(G)$, the following
holds\footnote{See~\cite{chung97} for the spectral implications of
  this definition of an expander.}
\[
|E(G) \cap (S \times \bar{S})| \ge \alpha \min\left(\vol_G(S), \vol_G(\bar{S})\right).
\]

Additionally, in our arguments, we need to argue about vertices of $H$ and $G$ with similar degrees. We choose the following definition
\begin{definition} 
For any graph $X$, vertices $u,v \in X$ with $|\Gamma_X(u)| \geq |\Gamma_X(v)|$ have \emph{$\eps$-similar degree} when 
$|\Gamma_X(u)| - |\Gamma_X(v)|  \leq 2 \epsilon |\Gamma_X(u)|.$  
\end{definition}

\hl{Frequently, we refer to the intersection of two vertices.
Let the set of vertices in the intersection of $u$ and $v$ in graph $X$ be denoted by
$I_X(u,v) = \Gamma_X(u) \cap \Gamma_X(v)$, 
and similarly let the intersection of $u,v,w$ 
in graph $X$ be $I_X(u,v,w) = \Gamma_X(u) \cap \Gamma_X(v) \cap \Gamma_X(w)$.}

\hl{We denote the vertices of $K_3$ as $\{a,b,c\},$ 
and often refer to these as \emph{colors}. 
To refer to an arbitrary color in $K_3$, we use the variable $t \in \{a,b,c\}$.}
For every vertex $v$ in $H$, there is a hidden pair of labels 
associated with $v$ from $P$.
Formally, $v$ corresponds to a node from the tensor $P$, 
whose vertices are tuples $(t,g) \in V(K_3) \times V(G)$.
We call $t$ the unknown color class of $v$ and $g$ the unknown $G$-class of $v$.
We use the same notation as in our informal problem statement,
and say that the function $\psi_{K_3}: V(H) \rightarrow V(K_3) = \{a,b,c\}$ extracts the unknown
color class of $v$,
and the function $\psi_{G}: V(H) \rightarrow V(G)$ extracts the unknown $G$ class of $v$. 
Then, we can define the complete hidden label function for $v \in V(H)$ as 
$\psi(v) = (\psi_{K_3}(v), \psi_G(v))$. 

\begin{definition}
For $v \in H$,  the functions $\psi_{K_3}: V(H) \rightarrow V(K_3) $ and $\psi_{G}: V(H) \rightarrow V(G)$ 
map $v$ to its unknown, underlying \emph{color class} and \emph{$G$-class} of $v$, respectively. 
Together, they give the hidden label
$\psi(v) = (\psi_{K_3}(v), \psi_G(v))$.  
\end{definition}

We refer to the hidden label function's inverse, $\psi^{-1}$, when discussing the unknown location of $(t,g) \in K_3 \times G=P$ 
in the graph $H$.
Instead of writing out $\{\psi(v_1), \psi(v_2),\ldots\}$ for a set $S = \{v_1,v_2,\ldots\}$, 
we use the shorthand $\psi(S)$, with similar notation for
$\psi^{-1}$. Note that $\psi$ is a bijection.

Next, we define a set of triples that will be important for our reconstruction algorithm.
\begin{definition}\label{def: core}
In the graph $H$, the \change{\emph{core triangles}}\footnote{\change{Note that the use of the word ``core" is unrelated to cores as defined by Hell and Ne{\v{s}}et{\v{r}}il \cite{hell1992core}. Also note 
the core triangles are not copies of $K_3$ in $H$.}} are the triples of the form 
\[\psi^{-1}(K_3 \times \{g\}) = \{\psi^{-1}(a,g), \psi^{-1}(b,g), \psi^{-1}(c,g)\}.\]
\end{definition}
One might wish to find the core triangles of all vertices in $G$ in order to prove our reconstruction goal, but such a task is impossible. 
In particular, nodes in $G$ may have such similar neighborhoods that 
it is impossible to distinguish vertices in $H$ with their $G$-classes.
This motivates our next definition.
\begin{definition}\label{def: confusable}
  Vertices $g,g' \in G$ are \emph{$\eps$-confusable} when
$ |I_G(g,g')| > (1-9\epsilon )\max\{|\Gamma_G(g)|,
|\Gamma_G(g')|\}.$ 
\end{definition}

Since the $\eps$-confusable vertices of $G$ make finding the core triangles impossible, we are content to reconstruct quasi-core triangles.
\begin{definition}\label{def: quasi-core}
A \change{set of 3 vertices} $\{v_1,v_2,v_3\} \subseteq H$ is a \emph{quasi-core} triangle
if the color classes of all three vertices are different and the $G$ classes of all three vertices are all $\eps$-confusable with each other, 
i.e. for $\psi_G(v_i)=g_i$, for $i \in [3]$, $g_1,g_2,g_3 $ are all $\eps$-confusable with each other.
\end{definition}
Note that core triangles are quasi-core.
Finding a set of disjoint quasi-core triangles that cover all nodes in $V(H)$ is challenging because of triangles that  
are not quasi-core, but may look similar to quasi-core triangles (see Proposition \ref{prop:C-edges}).
We call such triangles monochrome--as their nodes have the same color class--in order to distinguish them from quasi-core triangles. 
\begin{definition}\label{def: monochrome}
A \change{set of 3 vertices} $\{v_1,v_2,v_3 \}= \psi^{-1}(t \times \{g_1,g_2,g_3\}) \subseteq H$ 
\change{with $|\Gamma_H(v_1)| \geq |\Gamma_H(v_2)|,$ $|\Gamma_H(v_3)|$} is a \emph{monochrome} triangle if
\[
(1-8\epsilon) \cdot \frac{|\Gamma_H(v_1)|}{4}\leq |I_G(g_i,g_j)| \leq (1+9 \epsilon) \cdot \frac{|\Gamma_H(v_1)|}{4},
\] for all $i \neq j$.
\end{definition}

In fact, the only triples of vertices from $H$ that will be relevant to us are the quasi-core 
(which includes core) and monochrome triangles (see Lemma \ref{lem:tri-types}).

The following two propositions will be used very frequently in proving
Theorem \ref{thm:main}.  The first, Proposition~\ref{prop:error},
implies that $ |I_H(u,v)|$ is very close to
$|I_{K_3}(t_1,t_2)| \cdot |I_G(g_1,g_2)|$, and it follows from the
definition of the tensor structure, plus the constraint that no node
has more than an $\epsilon$ fraction of its edges removed to form $H$
from $P$.  Similarly, the second, Proposition~\ref{prop: bound-on-G},
says that for $u \in H$ with $\psi_G(u)=g$, $|\Gamma_G(g)|$ is very
close to $|\Gamma_H(u)|/2$.

\begin{proposition}\label{prop:error} 
  For $u,v \in H$ with $|\Gamma_H(u)|\geq |\Gamma_H(v)|$, let $\psi(u) = (t_1,g_1)$ and $\psi(v) =( t_2, g_2)$. 
  Then for $\epsilon \leq 1/3$,
  \[
    |I_{K_3}(t_1,t_2)| \cdot |I_G(g_1,g_2)| - 3\epsilon |\Gamma_H(u)| \leq |I_H(u,v)| \le |I_{K_3}(t_1,t_2)| \cdot |I_G(g_1,g_2)|.
  \]
\end{proposition}
\begin{proof}
From the definition of the tensor structure, we have that
\[
    |I_{K_3}(t_1,t_2)| \cdot |I_G(g_1,g_2)| =|I_P(u,v)|.
\] 
Further, deleting edges can never increase the size of an intersection in the graph, 
so $|I_H(u,v)|\leq |I_P(u,v)|$. 
The upper bound follows by combining the two equations. 
Since no node has more than an $\epsilon$ fraction of the edges \change{that} are deleted from $P$ to form $H$, 
we see that
\begin{align*}
  |I_H(u,v)| &\geq |I_P(u,v)|- \epsilon \left (|\Gamma_P(u)| +
               |\Gamma_P(v)| \right )\\
&\geq  |I_P(u,v)|- \frac{\epsilon}{1-\epsilon}\left (|\Gamma_H(u)| + |\Gamma_H(v)| \right )
\end{align*}
The lower bound follows since $|\Gamma_H(u)|\geq |\Gamma_H(v)|$, and $2\frac{\epsilon}{1-\eps} \leq3 \eps$ for $\eps \leq \frac13$.
\end{proof}

\begin{proposition}\label{prop: bound-on-G}
  For \change{any} vertex $u$ in $H$ with $\psi_G(u) = g$,
\begin{equation}
  \label{eq:bounded-G}
  \frac{|\Gamma_H(u)|}{2}\leq |\Gamma_G(g)| \leq  \frac{1}{1-\epsilon} \cdot\frac{|\Gamma_H(u)|}{2}.
\end{equation}
\end{proposition}
\begin{proof}
The lower bound on $|\Gamma_G(g)| $ follows 
from the observation that $|\Gamma_H(u)|\leq  |\Gamma_P(u)|= 2|\Gamma_G(g)| $
On the other hand, since the number of deletions from each vertex is
bounded, 
\[|\Gamma_H(u)| \ge (1-\eps)|\Gamma_P(u)| = 2 (1-\eps) |\Gamma_G(g)|.\]
\end{proof}

\subsection{Important graphs}\label{sec: graphs}

We define several graphs here that will be used throughout the proof section.
\paragraph{The candidate edge graph, $C$}

First, we construct a graph on $V(H)$,
whose edges contain those from the core triangles (see Proposition \ref{prop:core-contained}). 
Since our eventual goal is to find quasi-core triangles
and the edges in our graph are candidates for edges in the core triangles, 
we will call this graph the \emph{candidate edge} graph, $C$.
We define the graph $C$ on $V(H)$ such that distinct
$u, v \in V(H)$ \change{with $|\Gamma_H(u)| \geq |\Gamma_H(v)|$} form an edge $(u,v) \in E(C)$ if and only if
\begin{enumerate}
\item[(i)] $u$ and $v$ have $\eps$-similar degree in $H$, and
\item[(ii)]
  \begin{align}
    (1-6 \eps) \cdot \frac{|\Gamma_H(u)|}{2} \leq |I_H(u,v)| \leq
    \frac{1}{1-\epsilon} \cdot \frac{|\Gamma_H(u)|}{2}.\label{eq:j}
  \end{align}
\end{enumerate}

\change{In what follows, $C_i$ will denote a component of $C$.}

\begin{figure}
  \begin{center}
\includegraphics[width = 4cm]{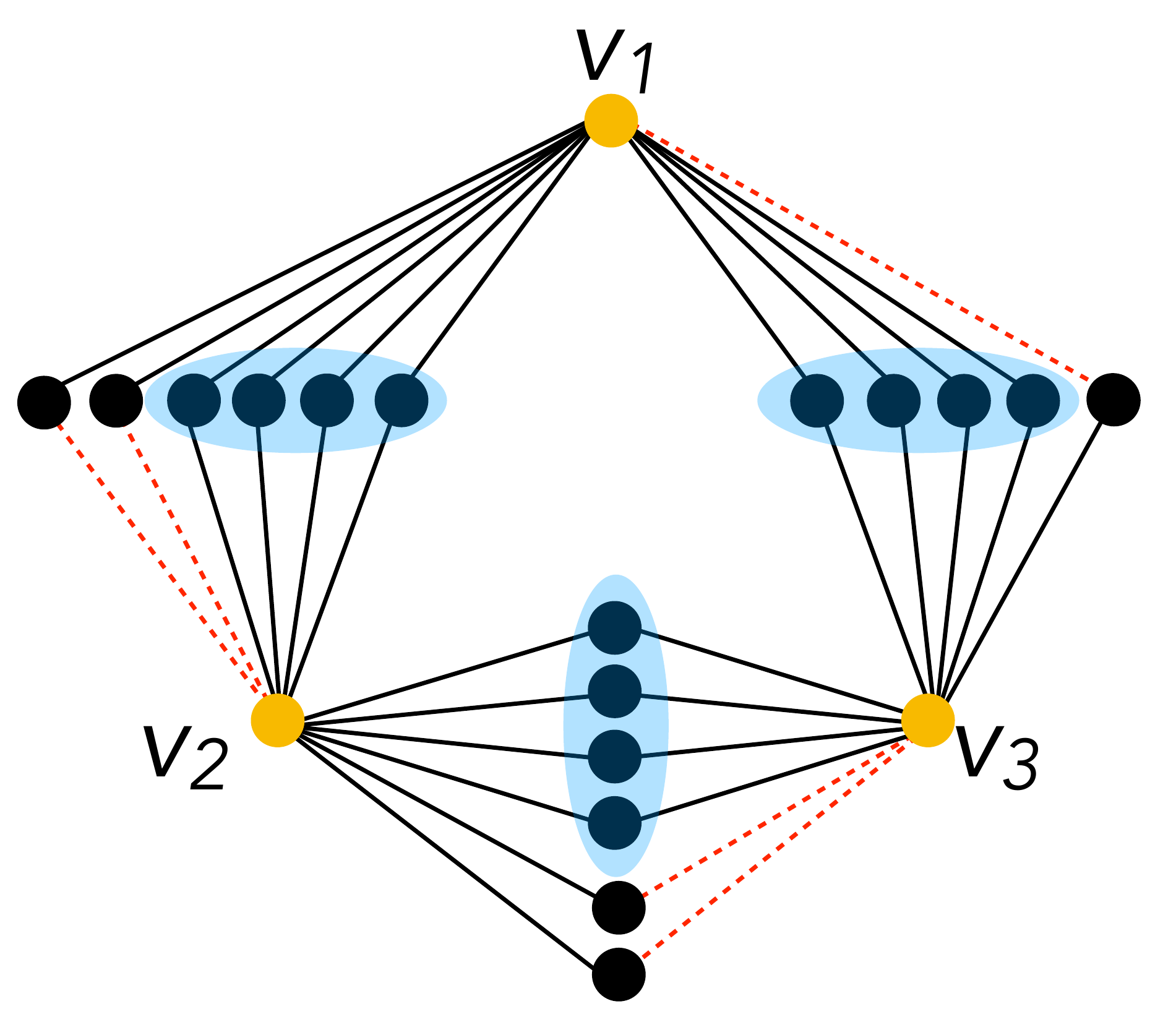}
  \end{center}
  \caption{An illustration of how the neighborhoods of three vertices in $C$ must look in order to form a triangle in $\mathcal{T}(C)$. 
  The dashed lines are edges deleted in building $H$ from $P$. 
}
\label{fig: neighborhood}
\end{figure}

\paragraph{The triangle graph, $\mathcal{T}(C)$}
Next, we construct a graph $\mathcal T(C)$, whose vertices are all
sets of triples from $V(C)$ such that for $v_1,v_2,v_3 \in V(C)$,
$(v_i,v_j) \in E(C)$ for all $i \neq j$ with $i,j \in \{1,2,3\}$. See
Figure~\ref{fig: neighborhood} for an illustration.
Note that by definition, since there are no self-loops in $C$,
$v_1,v_2,v_3$ are all distinct.  We will refer to the vertices of
$\mathcal{T}(C)$ as triangles.  An important property of $\mathcal{T}(C)$ is that it
includes all the core triangles (see Proposition
\ref{prop:core-contained}). We say the triangles
$T_1, T_2 \in \mathcal{T}(C)$ are adjacent if and only if they are
compatible in the following sense.

\begin{definition}\label{def:compatible}
Triangles $T_1, T_2 \in \mathcal{T}(C)$ are \emph{compatible}
exactly when
\begin{enumerate}
\item [(i)] $T_1$ and $T_2$ are on disjoint sets of vertices of $C$, and
\item [(ii)] there is an indexing $(u_1,u_2,u_3)$ of $T_1$ and
  $(v_1,v_2,v_3)$ of $T_2$ such that $(u_i,v_j) \in E(H)$ if and only
  if $i \neq j$.
\end{enumerate}
\end{definition}

The compatibility condition is defined so that triangles $T_1$ and $T_2$ corresponding to core triangles 
of $g$ and $g'$ in $\mathcal{T}$, i.e. $T_1=\psi^{-1}(K_3 \times \{g\})$ and $T_2=\psi^{-1}(K_3 \times \{g'\})$,
are adjacent in $\mathcal{T}$ if $(g,g') \in E(G)$ and none of the edges between $T_1$ and $T_2$
are deleted in building $H$ from $P$.

We say that two triangles of $\mathcal T(C)$ are \emph{connected} if
there is a path of compatible triangles from one to the
other. \change{Let $Y_j$ denote a connected component of
$\mathcal{T}(C)$}.  Note that a
component $Y_j$ of $\mathcal{T}(C)$ might intersect many components of
$C$, because the triangles of $Y_j$ are connected via edges of $H$,
which may cross different components of $C$.  We will show in Lemma
\ref{lem:discrete-covering} that for any $Y_j$ and any $C_i$, $Y_j$
either contains every vertex in $C_i$ or none of them.

\section{Algorithmic Results}\label{sec: proof-alg}

In this section, we provide formal algorithm statements and proofs for our algorithmic results.
Through out, we take $\eps_0 =1/40$.

\subsection{Properties of $C$, the Graph of Candidate Edges}
We use the graph $C$ to identify the core triangles, or triangles that
are close to the core triangles.  The following \change{proposition} shows that the
edges within core triangles will be kept in $C$.
\begin{proposition}\label{prop:core-contained}
\change{Fix $\epsilon < \eps_0.$ For all $g \in G$, all pairs of distinct vertices in $\psi^{-1}(K_3 \times \{g\} )$ 
contain an edge in $C$.}
\end{proposition}
\begin{proof}
Fix distinct $u,v \in C$ with $\psi(u) = (a,g)$ and $\psi(v) =( b,g)$. 
By rewriting the equation in Proposition \ref{prop: bound-on-G},
both $u$ and $v$ have neighborhood sizes satisfying
\[(1-\epsilon)\cdot 2 |\Gamma_G(g)| \leq |\Gamma_H(u)|, |\Gamma_H(v)| \leq 2 |\Gamma_G(g)|.\]
Therefore, $u$ and $v$ have $\eps$-similar degree in $H$.

Now, we show $|I_H(u,v)|$ falls in the specified range for 
$(u,v)$ to be an edge in $C$.
Using Proposition \ref{prop:error}, 
set $t_1=a, t_2=b,$ and $g_1=g_2=g$ to see that
\begin{align*}
  |I_{K_3}(a,b)| \cdot |I_G(g,g)| - 3\epsilon |\Gamma_H(u)| &\leq |I_H(u,v)| \le |I_{K_3}(a,b)| \cdot |I_G(g,g)\change{|}
\end{align*}
\change{Since $|I_G(g,g)| =|I_G(g)| $ and $  |I_{K_3}(a,b)| =1$, the above simplifies to
\begin{align*}
  |\Gamma_G(g)|- 3\epsilon |\Gamma_H(u)| &\leq |I_H(u,v)| \le  |\Gamma_G(g)|.
\end{align*}}
We can further lower bound $|I_H(u,v)|$ by using the fact that 
$|\Gamma_H(u)| \leq 2|\Gamma_G(g)|$ from Proposition \ref{prop: bound-on-G}, so
\begin{equation}\label{eqn: core-big-int}
\left (1-6 \eps\right ) \cdot |\Gamma_H(u)|/2 \leq  |I_H(u,v)|.
\end{equation}

On the other hand, $|\Gamma_H(u)| \geq 2 (1-\eps) |\Gamma_G(g)|$
from the lower bound in Proposition \ref{prop: bound-on-G}.
Using this as an upper bound on $|\Gamma_G(g)|$, we see that 
$|I_H(u,v)| \leq |\Gamma_H(u)|/ (2(1-\epsilon))$.
Comparing with Equation \ref{eq:j}, it follows that $(u,v)$ is an edge in $C$.
\end{proof}

Next, Proposition~\ref{prop:C-edges} \change{together with Proposition \ref{prop:error}} proves that edges in $C$ either come from 
(1) vertices whose unknown color classes are the same and whose $G$ classes
have intersection \change{size roughly half that of their neighborhoods in $G$}
or (2) vertices whose unknown color classes are different and whose $G$ classes have almost identical intersection.
The edges from core triangles are of \change{type (2)}.

\begin{proposition}\label{prop:C-edges}
  Fix $\epsilon < \eps_0.$ For $(u,v) \in C$, let $\psi(u) = (t_1,g_1)$ and $\psi(v) = (t_2,g_2)$, 
where $|\Gamma_H(u)| \geq |\Gamma_H(v)|$. 
Either (1) $t_1=t_2$ and 
\[
(1-6 \epsilon)\cdot \frac{|\Gamma_H(u)|}{4}\leq |I_G(g_1,g_2)| \leq (1+ 8\epsilon) \cdot \frac{|\Gamma_H(u)|}{4}.
\]
or
(2) $t_1 \neq t_2$ and 
\[
(1-6\epsilon) \cdot \frac{|\Gamma_H(u)|}{2}\leq |I_G(g_1,g_2)| \leq (1+8 \epsilon) \cdot \frac{|\Gamma_H(u)|}{2}.
\]
\end{proposition}

\begin{proof}
If $t_1=t_2$, then by Proposition \ref{prop:error} 
\[
 2\cdot |I_G(g_1,g_2)| - 3 \eps |\Gamma_H(u)| \leq  |I_H(u,v)| \leq  2 \cdot|I_G(g_1,g_2)|,
\]
and recall from Equation \ref{eq:j} that
\[(1-6\epsilon) \cdot \frac{|\Gamma_H(u)|}{2} \leq |I_H(u,v)| \leq \frac{1}{1-\epsilon} \cdot \frac{|\Gamma_H(u)|}{2}.\]
Crossing the upper and lower bounds from the inequalities, we see that
\[
 (1-6 \epsilon)\cdot \frac{|\Gamma_H(u)|}{4}\leq |I_G(g_1,g_2)| \leq\left ( \frac{1}{1-\epsilon} +6\epsilon\right ) \cdot \frac{|\Gamma_H(u)|}{4}.
\]
The same argument holds when $t_1 \neq t_2$, 
but \change{when applying Proposition \ref{prop:error}, the factor $|I_{K_3}(t_1,t_2)|$} is now a 1 instead of a 2, 
so we obtain
\[
(1-6 \epsilon) \cdot \frac{|\Gamma_H(u)|}{2}\leq |I_G(g_1,g_2)| \leq\left ( \frac{1}{1-\epsilon} +6\epsilon\right ) \cdot \frac{|\Gamma_H(u)|}{2}.
\]
Then we see that the statement holds, since for $\eps < \eps_0$
\[\frac{1}{1-\epsilon} +6\epsilon \leq 1+8 \eps.\]
\end{proof}

Next, we consider certain triples in $C$ with the graph $\mathcal{T}(C)$. 
Recall that the triples in $\mathcal{T}(C)$ have an identifiable structure that is found in the core triangles 

\subsection{Properties of $\mathcal{T}(C)$, the Graph of Triangles of $C$}

Recall \change{ we let $C_i$ denote a component of $C$.}
\hl{We will show that
one can use the partitioning induced by the components of $C$ to
partition the core triangles.  Further, if $\psi^{-1}(t,g)$ and
$\psi^{-1}(t',g')$ are in the same component of $C$, then they remain
in the same component of $\mathcal{T}(C)$, i.e.,
$\psi^{-1}(K_3 \times \{g\})$ and $\psi^{-1}(K_3 \times \{g'\})$ are
in the same component in $\mathcal{T}(C)$, as shown in the following
lemma.}

\begin{lemma}\label{lem:unbroken-comp}
 \hl{Fix $\epsilon < \eps_0.$ If $\psi(v)=(t,g)$ and $\psi(v')=(t',g')$
  are in the same component of $C$, then
  $\psi^{-1}(K_3 \times \{g\})$ and $\psi^{-1}(K_3 \times \{g'\})$ are
  in the same component in $\mathcal{T}(C)$.}
\end{lemma}
\begin{proof}

Suppose for now that $v$ and $v'$ are incident in component $C_i$ of $C$.
We want to show that there is some $g'' \in I_G(g,g')$ 
where $\psi^{-1}(K_3 \times \{g''\})$ is compatible with both
$\psi^{-1}(K_3 \times \{g\})$ and $\psi^{-1}(K_3 \times \{g'\}).$  
Such a $g''$ guarantees that triangle 
$\psi^{-1}(K_3 \times \{g''\})$ is incident to both 
triangles $\psi^{-1}(K_3 \times \{g'\})$  and $\psi^{-1}(K_3 \times \{g\})$  in $\mathcal{T}(C),$
proving they are in the same component of $\mathcal{T}(C)$.
We have a lower bound on $|I_G(g,g')|$ from 
Proposition \ref{prop:C-edges}, 
so it remains to argue that not too many compatibility relations are destroyed
in the process of deleting edges from $P = K_3 \times G$ to form $H$.

Recall that for any node $u$ in $\psi^{-1}(K_3 \times \{g\})$,
$|\Gamma_H(u)| \leq 2 |\Gamma_G(g)|$.
Using this and the fact that at most an $\epsilon$ fraction of edges are deleted from any vertex,
there are at most $6 \epsilon \cdot |\Gamma_G(g)|$ edges deleted from vertices 
$\psi^{-1}(K_3 \times \{g\})$ that would connect to core triangles.
That implies that out of the $ |I_G(g,g')|$ possible core triangles $\psi^{-1}(K_3 \times \{g''\})$ for $g'' \in I_G(g,g')$, 
at least $|I_G(g,g')|-6 \epsilon |\Gamma_G(g)|$ of them still maintain all 6 edges between  $\psi^{-1}(K_3 \times \{g''\})$ and 
$\psi^{-1}(K_3 \times \{g\})$ that are necessary for the triangles to be compatible.
This similarly holds for $\psi^{-1}(K_3 \times \{g'\})$,
at least $|I_G(g,g')|-6 \epsilon |\Gamma_G(g')|$ of them still maintain all 6 edges between  $\psi^{-1}(K_3 \times \{g''\})$ and 
$\psi^{-1}(K_3 \times \{g'\})$.

Then from Proposition
\ref{prop:C-edges}, we have that
\[(1-6\epsilon) \cdot \max(|\Gamma_H(v)|, |\Gamma_H(v')|)/4\leq |I_G(g,g')|.\]
 Since for $u \in \{v,v'\}$, $|\Gamma_H(u)| \geq 2 |\Gamma_G(\psi_G(u))|(1-\epsilon)$ by Proposition \ref{prop: bound-on-G},
it follows that
\[(1-6 \epsilon)\cdot 2\max(|\Gamma_G(g)|, |\Gamma_G(g')|) \cdot (1-\epsilon)/4\leq |I_G(g,g')|.\]
Therefore, for $g'' \in I_G(g,g')$, the number of triangles
$\psi^{-1}(K_3 \times \{g''\})$ 
in $\mathcal{T}(C)$ that have the 12 required edges to be compatible with both 
$\psi^{-1}(K_3 \times \{g\})$ and $\psi^{-1}(K_3 \times \{g'\})$ is at least
\[ ((1-6\epsilon)(1-\epsilon)/2- 12 \epsilon  ) \max(|\Gamma_G(g)| ,|\Gamma_G(g')|),\]
which is positive since $(1-6\epsilon)(1-\epsilon)/2 - 12 \epsilon >0$
for $\eps < \eps_0.$

Now suppose that $v$ and $v'$ are not necessarily incident in $C_i$. 
There is some path of nodes in $C_i$ connecting them, 
and every pair of incident vertices in that path is in the same component of $\mathcal{T}(C)$.
Transitively, it must be that $v$ and $v'$ are also in the same component of $\mathcal{T}(C)$. 
\end{proof}

Next, we discuss the structure of $\mathcal{T}(C)$ in more depth. 
In particular, we will look at which triples of nodes from $C$ become vertices in $\mathcal{T}(C)$, 
and we will study the edge relations between the triangles of $\mathcal{T}(C)$. 
Recall that the triples we find may be quasi-core triangles and not necessarily core triangles
(see Definition \ref{def: quasi-core}). 
Let  $\psi(v_1) = (a,g_1)$, $\psi(v_2)=(b,g_2)$, and $\psi(v_3) = (c,g_3)$.
Overall, the edges could be deleted in such a way that makes it impossible to distinguish triangles 
$\psi^{-1}(K_3 \times \{g\})$ and $\psi^{-1}((a,g_1),(b,g_2),(c,g_3))$,
if the vertices $g,g_1,g_2,g_3$ are $\eps$-confusable (see Definition \ref{def: confusable}),
i.e. almost all of their neighbors in $G$ are in common.
Thus the need to allow quasi-core triangles in the reconstruction. 
For example, consider the case when $G$ is the $\epsilon$-noisy hypercube.
\change{This is the graph on vertices $\{0,1\}^\ell$, for $n = 2^\ell,$ where nodes are adjacent exactly when their hamming distance is $\eps \ell$.}
Here, every vertex $u$ in $G$ has neighbors whose neighborhoods
are all only an $\epsilon$ fraction different from $u$'s, 
so in $H$ it might be that $u$ has the same neighborhood as one--or many--of its neighbors.

\change{W}e have seen from Proposition \ref{prop:C-edges} that edges in $C$ can also be from nodes whose color 
classes are the same and whose $G$ classes have intersection size roughly $|\Gamma_H(v_1)|/4$, for
$|\Gamma_H(v_1)| \geq |\Gamma_H(v_2)|, |\Gamma_H(v_3)|$.
Such triangles $\{v_1,v_2,v_3 \}$ will be monochrome (see Definition \ref{def: monochrome}).
Now we can exactly describe all of the types of triangles in $\mathcal{T}(C)$.

\begin{lemma}\label{lem:tri-types}
Fix $\epsilon < \eps_0.$ \change{The vertices of $\mathcal{T}(C)$, which are all triangles of $C$,}
are all quasi-core or monochrome triangles.
\end{lemma}
\begin{proof}
Fix a triangle $T = \{v_1,v_2,v_3\}$ in $\mathcal{T}(C)$, where $\psi(v_i)=(t_i,g_i)$, for $i \in [3]$.
Without loss of generality that $|\Gamma_H(v_1)| \geq |\Gamma_H(v_2)|\geq |\Gamma_H(v_3)|$.  

First, suppose that $t_1 = t_2$. For sake of deriving a contradiction, suppose also that $t_3 \neq t_1,t_2$. 
By Proposition \ref{prop:C-edges},
\[
|I_G(g_1,g_2)| \leq (1+8 \epsilon )\cdot \frac{|\Gamma_H(v_1)|}{4}.
\]
We will show that the above upper bound cannot hold when $t_3 \neq t_1,t_2$ 
by  lower bounding $|I_G(g_1,g_2)|$ with $|I_G(g_1,g_2,g_3)|$, 
where $I_G(g_1,g_2,g_3)$ the intersection of $ \cap_{i \in [3]}\Gamma_G(g_i)$.

Again by Proposition \ref{prop:C-edges}, for there to be an edge between $v_1$ and $v_3$, 
as well as between $v_2$ and $v_3$, 
it must be that
\[|I_G(g_1,g_3)| \geq (1-6 \epsilon)  |\Gamma_H(v_1)|/2 \quad \text{and} \quad |I_G(g_2,g_3)| \geq (1-6 \epsilon) |\Gamma_H(v_2)| /2.\]
Since $v_1, v_2,v_3$  all have $\eps$-similar degree, $ |\Gamma_H(v_2)|, |\Gamma_H(v_3)| \geq (1-2\epsilon)  |\Gamma_H(v_1)|$.
Using the lower bounds on $|I_G(g_2,g_3)|$ and $|I_G(g_1,g_3)|$, the $\eps$-similar degrees property, \change{and Proposition \ref{prop: bound-on-G},}
we have that
\begin{align*}
|I_G(g_1,g_2,g_3)|  &\geq |\Gamma_G(g_3)| - |\Gamma_G(g_3) \setminus \Gamma_G(g_1)| - |\Gamma_G(g_3) \setminus \Gamma_G(g_2)|  \\
                    &\geq |\Gamma_G(g_3)| - (|\Gamma_G(g_3)| -(1-6  \epsilon)  |\Gamma_H(v_1)|/2) -  (|\Gamma_G(g_3)| -(1-6  \epsilon)  |\Gamma_H(v_2)|/2)\\
                    &\geq \frac{1-6 \epsilon}{2}   \cdot ( |\Gamma_H(v_1)| +|\Gamma_H(v_2)| )-  |\Gamma_G(g_3)| \\
                    &\geq \frac{(1-6 \epsilon) ( 2-2 \epsilon )}{2}   \cdot |\Gamma_H(v_1)|-  |\Gamma_H(v_1)|/(2(1-\eps)).
\end{align*}
Putting it all together,
\[
 (1+8\epsilon) \cdot \frac{|\Gamma_H(v_1)|}{4} \geq |I_G(g_1,g_2)| \geq |I_G(g_1,g_2,g_3)| \geq \left (\frac{(1-6 \epsilon) ( 2-2 \epsilon )}{2}   - \frac{1}{2(1-\eps)} \right )\cdot |\Gamma_H(v_1)|
\]
is false for all \change{$\eps < \eps_0$,} so we reach a contradiction.\footnote{The choice of $\eps_0=1/40$ comes from $\eps_0$ 
being the largest reasonable looking fraction that this holds for.}
Therefore, if two vertices of $C$ in a triangle of $\mathcal{T}(C)$ have the same color class, then 
the third vertex also has that same color class. 

\smallskip

It follows that all triangles consist of vertices in $C$ that either have the same color class, 
or all different color classes. 
For a triangle with vertices in $C$ having all the same color class, by Proposition \ref{prop:C-edges} 
and the fact that the degrees are $\eps$-similar, for all distinct $i, j \in [3]$,
\[
(1-6 \epsilon) (1-2\epsilon)\cdot \frac{|\Gamma_H(v_1)|}{4} \leq  |I_G(g_i,g_j)|\leq (1+8 \epsilon) \cdot \frac{|\Gamma_H(v_1)|}{4}.
\]
Given that
$(1-6\epsilon) (1-2\epsilon) \geq (1-8\epsilon)$ and $1+8 \eps \leq 1+9 \eps$ for all $\eps < \eps_0$,
we can compare this inequality on $ |I_G(g_i,g_j)|$ with that from the definition of monochrome (Definition \ref{def: monochrome}) 
to see that such a triangle is monochrome.

For a triangle with vertices having all different color classes, 
using Proposition \ref{prop:C-edges}, the fact that $g_1,g_2,g_3$ are \change{$\eps$-}similar, 
and Proposition \ref{prop: bound-on-G},
for all distinct $i, j \in [3] $
\[
(1-6 \epsilon)(1-2\epsilon)(1-\epsilon) \cdot|\Gamma_G(g_1)| \leq |I_G(g_i,g_j)|\leq (1+8  \epsilon)|\Gamma_G(g_1)|.
\]
Note that the reasoning for the above bound is exactly the same as when the vertices of the triangle all had the same color class, 
but the corresponding inequality in Proposition \ref{prop:C-edges} has an additional factor of 2 here, 
and the other additional factor of 2 (and $1-\eps$ in the lower bound)
come from using $\Gamma_G$ to bound instead of $\Gamma_H$.
It follows that such a triangle is quasi-core since in comparing with the definitions of quasi-core and $\eps$-confusable
(Definitions \ref{def: quasi-core}, \ref{def: confusable})
\[
  (1-6\epsilon)(1-2
\epsilon)(1-\epsilon) \geq (1-9  \epsilon).
\]
\end{proof}

\hl{The graph $\mathcal{T}(C)$ will have several components.
The next lemma shows that core triangles are only compatible with core or quasi-core triangles.}

\begin{lemma}\label{lem:components}
  Fix $\epsilon < \eps_0.$ For $T \in \mathcal T(C)$
  quasi-core, we have that $T$ is only compatible with other
  quasi-core triangles.  Consequently, any component of
  $\mathcal{T}(C)$ consists of only monochrome triangles or of only
  quasi-core triangles.
\end{lemma}
\begin{proof}
  Suppose that $T',T \in \mathcal{T}(C)$ are compatible triangles
  where $T=\{\psi^{-1}(a,g_1), \psi^{-1}(b,g_2),\\\psi^{-1}(c,g_3)\}$ is
  quasi-core and 
  $T' = \{\psi^{-1}(t_1,g'_1), \psi^{-1}(t_2,g'_2),
  \psi^{-1}(t_3,g'_3)\}$.  \change{Assume the vertices of $T$ and $T'$ are
  labeled so that Condition 2 from Definition \ref{def:compatible} holds;
  so $t_1 \neq b, c$, and
  $t_2 \neq a,c$, and $t_3 \neq a,b$. That is,
  $t_1 = a, t_2 = b, t_3 = c$.} By Lemma \ref{lem:tri-types},
  all triangles are either quasi-core or monochrome, so
  $T'$ is quasi-core, and every component consists of only monochrome
  triangles or only quasi-core triangles.
\end{proof}

Given that all components of $\mathcal{T}(C)$ consist of either monochrome triangles or quasi-core triangles,
we will refer to the former type of component as monochrome and the latter as core.
In particular, our choice to refer to the latter components as core and not quasi-core is purposeful, 
as the next lemma shows that any component of $\mathcal{T}(C)$ containing a triangle 
with the node $u$ also contains the core triangle that $u$ is a part of.

\begin{lemma}\label{lem:quasi-equals-core}
  Fix $\epsilon < \eps_0.$ Fix $T \in \mathcal{T}(C)$ in core
  component $Y_j$ of $\mathcal{T}(C)$.  For all $g = \psi_G(u)$ for
  some $u \in T$, the core triangle $\psi^{-1}(K_3 \times \{g\})$ is in
  component $Y_j$ of $\mathcal{T}(C)$.
\end{lemma}
\begin{proof}

  As in the statement of the lemma, fix $T \in \mathcal{T}(C)$ in component $Y_j$ of $\mathcal{T}(C)$. 
For any $u \in T$, take $g = \psi_G(u)$. We seek to show that
$\psi^{-1}(K_3 \times \{g\}) \in Y_j$.

  \change{
  
First off, if $T$ is core, there is nothing
more to prove, so suppose that $T$ is quasi-core, but not core.
Denote the vertices of $T$ as $u,v,w$, where without loss of
generality $\psi(u)=(a,g)$, $\psi(v)=(b,g')$, and $\psi(w)=(c,g'')$.
Overall, we will show that there is a core triangle
$\psi^{-1}(K_3 \times \{g^*\})$ that is disjoint from
$\psi^{-1}(K_3 \times \{g\})$ and $T$, but is also compatible with both
triangles. 

We give a few sentences describing our high-level strategy.  
Note that we can always consider the compatibility relation in Definition \ref{def:compatible} and $\psi^{-1}$ on the graph
$P = K_3 \times G$, as $H$ is formed from just removing edges in $P$.
As such, observe that in $P$, the triangle
$\psi^{-1}(K_3 \times \{g^*\})$ is compatible with (as in Definition \ref{def:compatible} but with respect to $P$) and vertex disjoint
from the copy of $T$ in $P$
and $\psi^{-1}(K_3 \times \{g\})$, when
$g^* \in I_G(g,g',g'')$.
When edges are removed to form $H$ from $P$, some triangles which were compatible are no longer.
It suffices to show that $| I_G(g,g',g'')|$ is large, i.e., bigger than
$(1-O(\epsilon)) \max\{\Gamma_G(g), \Gamma_G(g'), \Gamma_G(g'')\}$, and that
there exists at least one $\psi^{-1}(K_3 \times \{g^*\})$ 
whose edges required for the compatibility relation
remain \change{intact} in $H$.}
This proof follows very similarly to that of Lemma
\ref{lem:unbroken-comp}.

Since $T$ is quasi-core, the vertices $g,g',$ and $g''$ are all
$\eps$-confusable, i.e. they all have large pairwise intersection. 
If $|\Gamma_G(g)| \geq |\Gamma_G(g')| \geq |\Gamma_G(g'')| $,
then since $g,g',g''$ are $\epsilon$-confusable, 
$ |\Gamma_G(g)\setminus \Gamma_G(g')| \leq 9\eps |\Gamma_G(g)|$
and $ |\Gamma_G(g)\setminus \Gamma_G(\change{g''})| \leq 9\eps |\Gamma_G(g)|$.
Using these bounds, we can lower bound $|I_G(g,g',g'')|$ with
\begin{align*}
  |I_G(g,g',g'')| &\geq |\Gamma_G(g)|  - |\Gamma_G(g)\setminus \Gamma_G(g')| - |\Gamma_G(g)\setminus \Gamma_G(g'')|  \\
  |I_G(g,g',g'')| &\geq |\Gamma_G(g)|  - 18\eps|\Gamma_G(g)|  = (1-18\eps)|\Gamma_G(g)|.
\end{align*}
The choice of  $|\Gamma_G(g)|$ as the largest neighborhood was arbitrary
and only chosen to show that more generally, the following bound holds:
\[|I_G(g,g',g'')| \geq (1 - 18 \epsilon
)\max(|\Gamma_G(g)|, |\Gamma_G(g')|, |\Gamma_G(g'')|).\]

Now, we will show that not enough edges are deleted to destroy the necessary compatibility relations.
Since the number of edges in $P$ incident to $\psi^{-1}(K_3 \times \{g\})$ 
is $6 |\Gamma_G(g)|$,  
there are at most $6 \epsilon |\Gamma_G(g)|$ edges deleted between
$\psi^{-1}(K_3 \times \{g\})$ and another core triangle.
Thus, there are at least
\[(1 - 18 \epsilon)\max(|\Gamma_G(g)|, |\Gamma_G(g')|,
|\Gamma_G(g'')|)- 6 \epsilon|\Gamma_G(g)| \] 
core triangles
$\psi^{-1}(K_3 \times \{g^*\})$ that have $g^* \in I_G(g,g',g'')$ and
the 6 edges necessary for the compatibility relation to be preserved
between $\psi^{-1}(K_3 \times \{g^*\})$ and
$\psi^{-1}(K_3 \times \{g\})$. Let $J_g$ be the set of such $g^*$.

It suffices to prove that at least one $g^* \in J_g$ has that
$\psi^{-1}(K_3 \times \{g\})$ is compatible with $T$. Note that since
$g^{*} \in I_G(g,g',g'')$, every $g^* \in J_g$ is compatible with $T$
according to the edges of $P$. Thus, in order for
$\psi^{-1}(K_3 \times \{g\})$ to not be compatible with $T$ at least
one of the $6$ edges between them in $P$ must be deleted. However, at
most $6\eps\max(|\Gamma_G(g)|, |\Gamma_G(g')|, |\Gamma_G(g'')|)$
edges incident to $T$ were deleted from $P$ to get $H$. 

Therefore, the number of triangles of the form $\psi^{-1}(K_3 \times \{g^*\})$ that are compatible with both $T$ 
and $\psi^{-1}(K_3 \times \{g\})$ is at least
\[
  (1 - 30\eps)\max(|\Gamma_G(g)|, |\Gamma_G(g')|,
  |\Gamma_G(g'')|)
\]
We just require one such triangle to prove the statement, so it suffices that
$
  1 - 30 \epsilon >0$
to prove the claim, and this is true for all $\eps < \eps_0$.
\end{proof}

Now, putting together several of the previous lemmas, 
we will see that a core component of $\mathcal{T}(C)$ 
must either not intersect a component $C_i$ of $C$,
 or it must contain all the core triangles of $C_i$.

\begin{lemma}\label{lem:discrete-covering}
Fix $\epsilon < \eps_0.$ For any component $C_i$ of $C$ and any core component $Y_j$ of $\mathcal{T}(C)$, $Y_j$
either contains every core triangle of $C_i$ or does not contain any node from $C_i$.
\end{lemma}
\begin{proof}
  Fix a component $C_i$ of $C$, and suppose $Y_j$ is a core component
  of $\mathcal{T}(C)$ that contains a triangle with some node
  $v \in C_i$. By Lemma~\ref{lem:quasi-equals-core}, the core triangle
  $T$ with $v \in T$ must also be in $Y_j$. To finish, by
  Lemma~\ref{lem:unbroken-comp}, $Y_j$ contains every node of
  $C_i$ (and thus every triangle of $C_i$).
\end{proof}

\subsection{Structure of core components: Proof of Theorem~\ref{thm:no-sparse-cuts}}

We say that $X \subseteq V(H)$ is \emph{atomic} if for every component
$C_i$ of $C$, we have that $C_i \cap X = \emptyset$ or
$C_i \subseteq X$. 
\change{In other words, a set of vertices is atomic if it is a union of components.}
 For each component $Y_j$ of $\mathcal{T}(C)$, let
$U_j$ be the vertices of $C$ (and thus $H$) which are covered by at
least one triangle in $Y_j$. Note that $U_j$ is atomic by
Lemma~\ref{lem:quasi-equals-core}. Let $H[U_j]$ be the subgraph of $H$
induced by the vertices $U_j$. First, we show that if $Y_j$ is a core
component, then the cut between $U_j$ and $H \setminus U_j$ is sparse.

\begin{proposition}\label{prop:core-cut}
  Let $Y_j$ be a core component of $\mathcal{T}(C)$. Let $X \subseteq
  U_j$ and $Z \subseteq H \setminus U_j$ be atomic. Then,
  \[
    |E(H) \cap (X \times Z)| \le \frac{5\eps}{1-\eps}
    \min(\vol_H(X), \vol_H(Z)).
  \]
\end{proposition}
\begin{proof}
  \change{Recall that we denote the full product graph from which $H$ is constructed by $P = K_3 \times G$.}
  By Proposition~\ref{prop:core-contained}, if $X$ is atomic, then the
  vertices of $X$ can be expressed as a disjoint union of core
  triangles. That is, there exists \change{$\psi_G(X) \subseteq G$ such that
  $\psi(X) =  K_3 \times \psi_G(X)$.}
  Define \change{$\psi_G(Z)$ likewise, and note that $\psi_G(X)$ and $\psi_G(Z)$} must be
  disjoint. Since $Y_j$ is core, each core triangle of $X$ must not be
  compatible with each core triangle of $Z$. In particular, since $X$
  and $Z$ are atomic, for any edge $(g, g')$ between \change{$\psi_G(X)$ and $\psi_G(Z)$},
  at least one of the six edges between
  $(\psi^{-1}(t,g),\psi^{-1}(t',g'))$ must be in
  $E' = E(P) \setminus E(H)$. In particular, this implies that
  \begin{align*}
    |E(H) \cap (X \times Z)|
    &\le 5|E' \cap (X \times Z)|\\
    &= 5\min\left(\sum_{u \in X}
      |\Gamma_{E'}(u)
      \cap Z|, \sum_{u \in Z} |\Gamma_{E'}(u) \cap X|\right)\\
    &\le 5\eps \min \left(\sum_{u \in X} |\Gamma_P(u)|, \sum_{u \in Z}
      |\Gamma_P(u)|\right)\\
    &\le \frac{5\eps}{1-\eps} \min \left(\sum_{u \in X} |\Gamma_H(u)|, \sum_{u \in Z} |\Gamma_H(u)|\right)\\
    &\le \frac{5 \eps}{1-\eps} \min(\vol_H(X), \vol_H(Z)),
  \end{align*}
  as desired.
\end{proof}

We further show that if $Y_j$ is a core component, then
$H[U_j]$ can be colored by a simple algorithm.

\begin{proposition}\label{prop:core-color}
  Given a core component $Y_j$ of $\mathcal{T}(C)$, one can in $O(|Y_j|^2)$
  time find subsets
  $\mathcal A_j, \mathcal B_j, \mathcal C_j$ of $U_j$ such that
  $\mathcal A_j, \mathcal B_j,$ and $\mathcal C_j$ equal
  $\psi^{-1}_{K_3}(a) \cap U_j$, $\psi^{-1}_{K_3}(b) \cap U_j$, and
  $\psi^{-1}_{K_3}(c) \cap U_j$, up to a permutation of the three
  sets.
\end{proposition}

\begin{proof}

  Pick an arbitrary triangle $T_0 = \{u,v,w\} \in Y_j$. We choose to
  initialize
  $\mathcal A = \{u\}, \mathcal B = \{v\}, \mathcal C =
  \{w\}$. Initialize $S = \{T_0\}$. We now iterate the following
  procedure.
  
  Consider a triangle $T' \in Y_j \setminus S$ which is compatible
  with some $T \in S$. Assume that $T = \{u_a,v_b,w_c\}$, with
  $u_a \in \mathcal A$, $v_b \in \mathcal B$ and $w_c \in \mathcal
  C$. Since $T$ and $T'$ are compatible, there is a unique vertex
  $u' \in T'$ which is non-adjacent to $u_a$ in $H$. We place this
  vertex in $\mathcal A$. Likewise, we place the unique $v' \in T'$
  non-adjacent to $v_b$ in $\mathcal B$ and the unique $w' \in T'$
  non-adjacent to $w_c$ in $\mathcal C$. Finally, we add $T'$ to
  $S$. This procedure will iterate until $S = Y_j$ because $Y_j$ is
  connected.

  It suffices to show that if $Y_j$ is core, then
  $\mathcal A, \mathcal B, \mathcal C$ will respect the color classes
  induced by $\psi^{-1}_{K_3}$. To see why, each triangle $T$ of $Y_j$
  is quasi-core by Lemma~\ref{lem:components}, so its three vertices
  land in the different color classes induced by
  $\psi^{-1}_{K_3}$. Further, for any pair $T, T'$ of compatible
  triangles the non-adjacent pairs of vertices must lie in the same
  color class. Thus, at every step of the algorithm, the sets
  $\mathcal A, \mathcal B, \mathcal C$ are consistent with
  $\psi^{-1}_{K_3}$. Therefore, the final
  $\mathcal A, \mathcal B, \mathcal C$ must precisely be
  $\psi^{-1}_{K_3}(a) \cap U_j$, $\psi^{-1}_{K_3}(b) \cap U_j$, and
  $\psi^{-1}_{K_3}(c) \cap U_j$, up to a permutation of the three
  sets.
\end{proof}

We now have the tools to prove Theorem \ref{thm:no-sparse-cuts}.

\begin{proof}[Proof of Theorem \ref{thm:no-sparse-cuts}]
We claim that if $G$ is a $3\eps$-edge-expander, then there is a
single core component $Y_j \in \mathcal T(C)$.
\change{Therefore, we can loop over all the $Y_j$'s (which can be done efficiently)
and output one that gives a full 3-coloring of $H$.
There exists a core component by Lemma \ref{lem:components}
and we can consider one such component as input to Proposition~\ref{prop:core-color}.
As we loop over the $Y_j$'s, we will encounter components that are not core,
but the procedure in the proof of Proposition~\ref{prop:core-color} can still be applied on each component, 
where if a 3-coloring is not produced on $H$ as a result, then that component is definitively not core.}

Consider a core component $Y_j$. Let $G_j$ be the vertices of $G$
which correspond to vertices of $U_j$. Since $G$ is a $3\eps$-edge-expander,
we know that
\[
  |E(G) \cap (G_j \times (G \setminus G_j))| \ge 3\eps
  \min(\vol_G(G_j), \vol_G(G\setminus G_j)),
\]
Thus, since each edge of $G$ becomes $6$ edges in $P$, we have that
\begin{align*}
  |E(P) \cap (U_j \times (V \setminus U_j))|
  &= 6|E(G) \cap (G_j \times (G \setminus G_j))|\\
  &\ge 18\eps
  \min(\vol_G(G_j), \vol_G(G\setminus G_j))\\
  &= 9\eps\min(\vol_P(U_j), \vol_P(V \setminus U_j))\\
  &\ge 9\eps\min(\vol_H(U_j), \vol_H(V \setminus U_j)).
\end{align*}
Also, note that by the definition of $E'$ (the edges deleted from $P$), \[|E'
\cap (U_j \times (P \setminus U_j))| \le \eps \min(\vol_P(U_j),
\vol_P(V \setminus U_j)) \le \frac{\eps}{1-\eps} \min(\vol_H(U_j),
\vol_H(V \setminus U_j)).\]
Thus, 
\[
|E(H) \cap (U_j \times (V \setminus U_j))| \ge \left(9 - \frac{1}{1-\eps}\right)\eps\min(\vol_H(U_j), \vol_H(V \setminus U_j))
\]

Since $U_j$ is atomic by Lemma \ref{lem:quasi-equals-core}, 
we can use Proposition~\ref{prop:core-cut}
to see that
\begin{align*}
  |E(H) \cap (U_j \times (H \setminus U_j))| \le \frac{5\eps}{1-\eps}
    \min(\vol_H(U_j), \vol_H(H \setminus U_j)).
\end{align*}
Thus, since $9-\frac{1}{1-\eps} > \frac{5 \eps}{1-\eps}$ 
for $\eps < \eps_0$, we must have that
$\min(\vol_H(U_j), \vol_H(H \setminus U_j)) = 0$.
\change{Every vertex of $H$ has at least one edge coming from it,
since in constructing $H$ at least a $(1-\epsilon)$ fraction of edges were preserved from the underlying tensor, 
and therefore $Y_j$ does
indeed cover all of $H$, as desired.}
\end{proof}

\subsection{Core Factoring Algorithm}

Leading up to the proof of Theorem~\ref{thm:main}, we now use Proposition~\ref{prop:core-color} to prove our Core
Factoring Lemma
(Lemma~\ref{lem:quasi-core-match}). Algorithm~\ref{alg:core-factoring}
outlines our approach to Lemma~\ref{lem:quasi-core-match}.
The proof of Lemma~\ref{lem:quasi-core-match} also relies on 
Propositions \ref{prop:deg}, \ref{prop:3-int}, and \ref{prop:link}, 
\hl{which are stated within the proof} of Lemma~\ref{lem:quasi-core-match}.

\change{For the algorithm, we define a matching on a weighted bipartite graph 
as the \emph{max-min matching} (also known as a \emph{bottleneck} matching)
to be the matching that
maximizes the minimum weight of any edge in the matching (e.g., \cite{punnen1994improved} and references therein).
One efficient algorithm for max-min matching 
is to sort the edges by weight
and iteratively delete the smallest one until no matching
remains. The last edge deleted is the objective value. This can be made more efficient with a binary search (c.f., \cite{bottleneck}).}

\begin{algorithm}[ht]
  \caption{ Core Factoring Algorithm}\label{alg:core-factoring}
  \begin{algorithmic}
    \small \State \textbf{Input:} A component $Y_j$ of
    $\mathcal{T}(C)$ and $U \subseteq U_j$ atomic.
    \State \textbf{Output:} \textsf{FAIL} or a
    factorization $K_3 \times \widetilde{G}_U$ of
    $\widetilde{H}_U$ on vertex set $U$ with maps
    $\psi^{U}_{K_3} : U \to K_3$ and
    $\psi^{U}_{G} : U \to \widetilde{G}_U$
    \State
\State Compute $\mathcal A_j, \mathcal B_j, \mathcal C_j \subseteq H[U_j]$
according to Proposition~\ref{prop:core-color}.
\State Let $\mathcal A = \mathcal A_j \cap U$, $\mathcal B = \mathcal
B_j \cap U$, $\mathcal C = \mathcal C_j \cap U$
\If{$\mathcal A, \mathcal B, \mathcal C$ have unequal sizes or is not a 3-coloring of
$H[U]$} \Return \textsf{FAIL} \EndIf
\State \hspace{-1cm} \textit{ // Max-min matching phase}
\State Construct graph
$(\mathcal{A}\dot{\cup }\mathcal{B}\dot{\cup } \mathcal{C}, E)$,
where $E = (\mathcal{A} \times \mathcal{B}) \cup (\mathcal B
\times \mathcal C)$.

\For{Each edge $(u,v)$ of $E$} assign weight
$2|I_H(u,v)| / \max\{|\Gamma_H(u)|, |\Gamma_H(v)|\}$. \EndFor

\State Find a max-min matching $\mathcal{M}_1$ between $\mathcal{A}$
and $\mathcal{B}$

\State Find a max-min matching $\mathcal{M}_2$
between $\mathcal{B}$ and $\mathcal{C}$

\If{objective value of either $\mathcal{M}_1$ or $\mathcal{M}_2$ is
less than $1 - 6\eps$} \Return \textsf{FAIL} \EndIf

\State \hspace{-1cm} \textit{ // Building the tensor decomposition}

\State Initialize $V(\widetilde{G}_U) = \{g_v : v \in \mathcal B\}$ and $V(\widetilde{H}_U) = U$.

\For{each $v$ in $\mathcal{B}$} \State Let $u \in \mathcal A$ be
unique $(u,v) \in \mathcal{M}_1$, $w \in \mathcal C$ be unique
$(v,w) \in \mathcal M_1$ \State Set $\psi_{K_3}^{U}(u) = a$,
$\psi_{K_3}^{U}(v) = b$, $\psi_{K_3}^{U}(w) = c$.  \State Set
$\psi_{G}^{U}(u) = \psi_G^{U}(v) = \psi_G^{U}(w) = g_v$.
\For{each $v'$ in $\mathcal B$} \State Let $u'$ be unique
$(u',v') \in \mathcal M_1$, $w'$ be unique $(v',w') \in \mathcal
M_2$. \If{any of $(u,v'),(u,w'),(v,u'),(v,w'),(w,u'),(w,v') \in H$}
\State Add $(g_v,g_{v'})$ to $E(\widetilde{G}_U)$ \State Add
$(u,v'),(u,w'),(v,u'),(v,w'),(w,u'),(w,v')$ to $E(\widetilde{H}_U)$
\EndIf \EndFor \EndFor \If{$|E(H[U]) \Delta E(\widetilde{H}_U)| \le
     260\eps |E(H[U])| + |U|$} \Return
$\widetilde{G}_U$,$\widetilde{H}_U$, $\psi_{K_3}^{U}$,
$\psi_G^{U}$ \Else\ \Return \textsf{FAIL}\EndIf
  \end{algorithmic}
\end{algorithm}

Recall for the next lemma that $U_j$ are the vertices of $H$ that 
are in some triangle in $Y_j$.
\begin{lemma}[Core Factoring Lemma]\label{lem:quasi-core-match}
  Fix $\epsilon < \eps_0.$ Fix $Y_j$ a component of $\mathcal
  T(C)$. Let $U$ be an atomic subset of $U_j$.
  In $O(|Y_j|^3)$ time, \change{Algorithm \ref{alg:core-factoring} outputs} either
  \textsf{FAIL} or a graph $\widetilde{H}_U$ on $U$
  with a factorization $K_3 \times \widetilde{G}_U$, as given by maps
  $\psi^{U}_{K_3} : U \to K_3$ and $\psi^{U}_{G} : U \to \widetilde{G}_U$
  with the following properties:
  \begin{itemize}
   \item[(a)] $\psi^{U}_{K_3}$ is a 3-coloring of $H[U]$.
   \item[(b)] $|E(H[U]) \Delta E(\widetilde{H}_U)| \le
     260\eps \vol_H(U) + |U|$.
   \item[(c)] If $Y_j$ is a core component, the algorithm
     will never output \textsf{FAIL}.
   \end{itemize}
  \end{lemma}

  \begin{proof} The reader may use Algorithm \ref{alg:core-factoring}
    as an outline for this proof. Using
    Proposition~\ref{prop:core-color}, we can find
    $\mathcal A_j, \mathcal B_j, \mathcal C_j \subseteq U_j$ which
    cover every triangle of $Y_j$. Note that if $Y_j$ is core, then
    $|\mathcal A_j| = |\mathcal B_j| = |\mathcal C_j|$ by
    Lemma~\ref{lem:quasi-equals-core}, as $U_j$ can be partitioned by
    the core triangles of $Y_j$. Further, if $Y_j$ is core, then
    $\mathcal A_j, \mathcal B_j, \mathcal C_j$ induce a 3-coloring
    of $H[U_j]$. Let $\mathcal A = \mathcal A_j \cap U$,
    $\mathcal B = \mathcal B_j \cap U$, and
    $\mathcal C = \mathcal C_j \cap U$. Since $U$ is atomic,
    $\mathcal A_j, \mathcal B_j, \mathcal C_j$ each intersect each
    core triangle of $U$ in exactly one vertex. Thus, the partition from
    $\mathcal A, \mathcal B, \mathcal C$ induces a 3-coloring of
    $H[U]$ and each set has equal size. 
    We can safely report
    \textsf{FAIL} if one of these conditions fails to hold.

    \paragraph{Max-min matching phase} Consider the tripartite graph on $U$ partitioned into
    $\mathcal{A},\mathcal{B},$ and $\mathcal{C}$, with edge set
    $(\mathcal A \times \mathcal B) \cup (\mathcal B \times \mathcal C)$.
  For every edge $(u,v)$ in this tripartite graph, assign a weight
    to it of $2|I_H(u,v)| / |\Gamma_H(u)|$, for
    $|\Gamma_H(u)| \geq |\Gamma_H(v)|$.  
    We \hl{find a max-min}
    matching between sets $\mathcal{A}$ and
    $\mathcal{B}$,  $\mathcal M_1$, and between sets $\mathcal{B}$ and
    $\mathcal{C}$, $\mathcal M_2$. 
    Define the \emph{quality} of the \change{matchings} on
    $(\mathcal{A}, \mathcal{B},\mathcal{C})$ to be the minimum
    objective of the max-min matchings $\mathcal M_1$ and
    $\mathcal M_2$.

    We now show that if $Y_j$ is core, the \change{quality of the matchings} is at least
    $1-6\eps$. If node $u$ with $\psi(u) = (t,g)$ is in \change{$\mathcal{A}$,}
    then, by Lemma~\ref{lem:quasi-equals-core}, for distinct
    $t',t'' \neq t$, one of $\psi^{-1}(t',g), \psi^{-1}(t'',g)$ is in
    $\mathcal{B}$ and the other is in $\mathcal{C}$.  Letting $v$ be
    such that $\psi(v) \in \{(t',g), (t'',g)\}$, recall Equation
    \ref{eqn: core-big-int} states that
    \[\change{\left (1-6 \eps \right ) \cdot \max(|\Gamma_H(u)|,
    |\Gamma_H(v)|)/2 \leq |I_H(u,v)|}.\]  This implies that the quality
    of the matchings is at least $1-6 \epsilon$, since that value is
    achievable when the matching edges are all a subset of the core
    triangle edges . Thus, if the quality is less than $1-6\eps$, we
    may safely output \textsf{FAIL}, as $Y_j$ cannot be core.
  
    \paragraph{Building the tensor decomposition} It remains to build
    $\widetilde{G}_U$ and $\widetilde{H}_U$. For each
    $v \in \mathcal B$, let $g_v$ be a vertex of
    $V(\widetilde{G}_U)$. We let the vertices of $\widetilde{H}_U$ be
    $U$.
    For each $v \in \mathcal B$, let $T_v$ be the triangle
    $\{u,v,w\}$, where $u \in \mathcal A$ is such that $(u,v) \in
    \mathcal M_1$ and $w \in \mathcal C$ is such that $(v,w) \in
    \mathcal M_2$. By definition of $\mathcal M_1$ and $\mathcal M_2$,
    the $T_v$'s partition $U$.

    For every pair $v, v' \in \mathcal B$, we add $(g_v, g_{v'})$ to
    $E(\widetilde{G}_U$) if $T_v$ and $T_{v'}$ share at least one edge
    in $H$. Likewise, if $T_v$ and $T_{v'}$ share at least one edge in
    $H$, we add $(x,y)$ to $E(\widetilde{H}_U)$ for all $x \in T_v$
    and $y \in T_{v'}$ with $x$ and $y$ elements of distinct
    $\mathcal A, \mathcal B, \mathcal C$.  For each $u \in \mathcal A$, we define $\psi^{U}_{K_3}(u) =
    a$. Likewise, for each $v \in \mathcal B$, $\psi^{U}_{K_3}(v) = b$ and
    each $w \in \mathcal C$, $\psi^{U}_{K_3}(w) = c$. For each $v \in
    \mathcal B$, we define $\psi^{U}_G(x) = g_v$ for all $x \in T_v$. Note by definition
    $\widetilde{H}_U = K_3 \times \widetilde{G}_U$ as induced by
    $\psi_{K_3}^{U}$ and $\psi_{G}^{U}$. 

    \paragraph{Bounding reconstruction error} It is clear at this
    point that if the output does not \textsf{FAIL} then conditions (a),
    (b), (c) of the lemma hold. \change{Note in particular that condition (b) holds due to the line preceding the \textsf{else} statement.}
     It suffices to prove that the
    algorithm does not return \textsf{FAIL} when $Y_j$ is a core
    component. In particular, it suffices to check condition (b) when
    $Y_j$ is core. Observe that
    \begin{align*}
      |E(\widetilde{H}_U) \Delta E(H_U)|
      &= |E(\widetilde{H}_U) \setminus E(H_U)| + |E(H_U) \setminus E(\widetilde{H}_U)|.
    \end{align*}
    Note that $(x, y) \in E(H_U) \setminus E(\widetilde{H}_U)$ only if
    $x,y \in T_v$ for some $v \in \mathcal B$, as any edge between
    different $T_v$'s is accounted for in $\widetilde{H}_U$. Thus,
    $|E(H_U) \setminus E(\widetilde{H}_U)| \le 3|\mathcal B| = |U|.$

    Now, we bound $|E(\widetilde{H}_U) \setminus E(H_U)|$. 
    For simplicity of the remainder of the argument, we assume without
    loss of generality that $\mathcal A = \psi^{-1}_{K_3}(a) \cap U,
    \mathcal B = \psi^{-1}_{K_3}(b) \cap U,$ and $\mathcal C = \psi^{-1}_{K_3}(c) \cap U$
    (recall by Proposition~\ref{prop:core-color} that this is always true
    up to permutation of $\mathcal A, \mathcal B, \mathcal C$). We
    observe the following structural property about the matched triangles.

    \begin{proposition}\label{prop:deg}
      \change{Assume $Y_j$ is a core component, with $\mathcal A, \mathcal B,
      \mathcal C$ its color classes when restricted to $U$.} Fix
      $T_v = \{u \in \mathcal A,v \in \mathcal B,w \in \mathcal
      C\}$ found by the matching algorithm. Then,
      \[(1-14\eps)\max(|\Gamma_H(u)|, |\Gamma_H(v)|, |\Gamma_H(w)|)
        \le \min(|\Gamma_H(u)|, |\Gamma_H(v)|, |\Gamma_H(w)|).\]
    \end{proposition}
    
        \begin{proof}[Proof of Proposition~\ref{prop:deg}]
          Note that by the defining property of $\mathcal M_1$ and
          $\mathcal M_2$, we have that
          \begin{align}
            (1 - 6\eps) \max(|\Gamma_H(u)|, |\Gamma_H(v)|) / 2 &\le
                                                                 |I_H(u,v)|
            \label{eq:uv}\\
            (1 - 6\eps) \max(|\Gamma_H(v)|, |\Gamma_H(w)|) / 2 &\le
                                                                 |I_H(v,w)|\label{eq:vw}
          \end{align}
          Further, since $\psi_{K_3}(u) = a, \psi_{K_3}(v) = b,
          \psi_{K_3}(w) = c$, we have that any vertex $x \in I_H(u,v)$
          must have $\psi_{K_3}(x) = c$. Thus,
          \begin{align*}
            |I_H(u,v)| &\le \min(|\Gamma_H(u) \cap \psi^{-1}_{K_3}(c)|,
            |\Gamma_H(v) \cap \psi^{-1}_{K_3}(c)|)\\
                       &\le \min(|\Gamma_P(u) \cap \psi^{-1}_{K_3}(c)|,
                         |\Gamma_P(v) \cap \psi^{-1}_{K_3}(c)|)\\
            &= \frac{1}{2}\min(|\Gamma_P(u)|, |\Gamma_P(v)|)\\
                       &\le
            \frac{1}{2(1-\eps)} \min(|\Gamma_H(u)|, |\Gamma_H(v)|).
          \end{align*}
          We can further show that $2(1-\eps)|I_H(v,w)| \le
          \min(|\Gamma_H(v)|, |\Gamma_H(w)|)$. Thus, when combined with
          (\ref{eq:uv}) and (\ref{eq:vw}), we get that
          \begin{align*}
            (1 - 7\eps) \max(|\Gamma_H(u)|, |\Gamma_H(v)|) &\le
            \min(|\Gamma_H(u)|, |\Gamma_H(v)|),\\
            (1 - 7\eps) \max(|\Gamma_H(v)|, |\Gamma_H(w)|) &\le
            \min(|\Gamma_H(v)|, |\Gamma_H(w)|).
          \end{align*}
          By suitably combining these two inequalities, we can further
          deduce that \[(1-14\eps)\max(|\Gamma_H(u)|,|\Gamma_H(w)) \le
            \min(|\Gamma_H(u)|,|\Gamma_H(w)|).\]
          To see why the result follows, assume that $x \in \{u,v,w\}$
          has the highest degree and $y \in \{u,v,w\}$ has the lowest. If
          $x = y$, the result is obvious. Otherwise, we pick the
          inequality of the previous three featuring $x$ and $y$.
        \end{proof}
    
    Given vertices $g_1, g_2, g_3 \in G$, define their \emph{disjunction}
    $\Delta(g_1,g_2,g_3)$ to be
    \[
      \Delta(g_1,g_2,g_3) = |(\Gamma_G(g_1) \cup \Gamma_G(g_2) \cup
      \Gamma_G(g_3)) \setminus I_G(g_1,g_2,g_3)|
    \]
    We give a bound on the disjunction of the triangles we matched.

    \begin{proposition}\label{prop:3-int}
      Assume $Y_j$ is a core component, with $\mathcal A, \mathcal B,
      \mathcal C$ its color classes when restricted to $U$.
      Fix
      $T_v = \{u \in \mathcal A,v \in \mathcal B,w \in \mathcal
      C\}$. Let $g_u, g_v, g_w \in G$ be such that $u = \psi^{-1}(a,g_u), v =
      \psi^{-1}(b,g_v), w = \psi^{-1}(c,g_w)$. Then,
      \[
        |\Delta(g_u,g_v,g_w)| \le 50\eps\min(|\Gamma_H(u)|,|\Gamma_H(v)|,|\Gamma_H(w)|).
      \]
    \end{proposition}

    \begin{proof}[Proof of Proposition~\ref{prop:3-int}]
      Note that
      \[
        |\Delta(g_u,g_v,g_w)| \le |\Gamma_G(g_u)| + |\Gamma_G(g_v)| +
        |\Gamma_G(g_w)| - 3|I_G(g_u, g_v, g_w)|.
      \]
      Further,
      \[
        |I_G(g_u, g_v, g_w)| \ge |\Gamma_G(g_v)| - |\Gamma_G(g_v)
        \setminus I_G(g_u,g_v)| - |\Gamma_G(g_v) \setminus I_G(g_v, g_w)|. 
      \]
      See that via (\ref{eq:uv})
      \begin{align*}
        |\Gamma_G(g_v) \setminus I_G(g_u,g_v)| &=
        \frac{1}{2}|\Gamma_P(v)| - |I_P(u, v)|\\ &\le
        \frac{1}{2(1-\eps)}|\Gamma_H(v)| - |I_H(u,v)|\\ &\le
        \left(\frac{1}{2-2\eps} -
                                                          \frac{1-6\eps}{2}\right)|\Gamma_H(v)|\\
                                               &\le 4\eps |\Gamma_H(v)|,
      \end{align*}
      where the last line follows for $\eps < \eps_0$.
      Thus, working out the symmetric bound for $|\Gamma_G(g_v)
      \setminus I_g(g_v, g_w)|$, we have that
      \begin{align*}
        |\Delta(g_u,g_v,g_w)| &\le |\Gamma_G(g_u)| + |\Gamma_G(g_w)| -
        2|\Gamma_G(g_v)| + 24 \eps |\Gamma_H(v)|\\
        &\le \frac{1}{2(1-\eps)} |\Gamma_H(u)| +
          \frac{1}{2(1-\eps)}|\Gamma_H(w)| - (1 - 24\eps) |\Gamma_H(v)|.
      \end{align*}
      Further using Proposition~\ref{prop:deg}, we get that
      \begin{align*}
        |\Delta(g_u,g_v,g_w)| &\le
        \frac{1}{(1-\eps)(1-14\eps)}\min(|\Gamma_H(u)|,|\Gamma_H(v)|,|\Gamma_H(w)|)\\
        &\ \ \ \ \ \ \ \ \ - (1 - 24\eps)\min|\Gamma_H(u)|,|\Gamma_H(v)|,|\Gamma_H(w)|)\\ &\leq 50 \eps\min(|\Gamma_H(u)|,|\Gamma_H(v)|,|\Gamma_H(w)|),
      \end{align*}
      for $\eps < \eps_0$, as desired.
    \end{proof}

    Recall from Definition \ref{def:compatible} that triangles
    $T_v = \{u, v, w\}$ and $T_{v'} = \{u',v',w'\}$ (with
    $u,u' \in \mathcal A,$ $v,v' \in \mathcal B$,
    $w,w'\in \mathcal C$) are compatible if
    \[(u,v'),(u,w'),(v,u'),(v,w'),(w,u'),(w,v') \in E(H).\] 
    We will call triangles
    $T_v$ and $T_{v'}$ \emph{weakly linked} if at least one but
    not all of those six pairs is an edge of $H$. 
    The edges in $E(\widetilde{H}_U) \setminus E(H_U)$ come from 
    weakly linked triangles. 
    Thus, \hl{we can upper bound} $|E(\widetilde{H}_U) \setminus E(H_U)|$ with
    5 times the number of weakly linked triangles, 
    as each pair of weakly linked triangles contributes at most 5 edges to $E(\widetilde{H}_U) \setminus E(H_U)$.

    \begin{proposition}\label{prop:link}
      Assume $Y_j$ is a core component, with
      $\mathcal A, \mathcal B, \mathcal C$ its color classes when
      restricted to $U$.  There are at most $52\eps \vol_H(U)$ weakly
      linked pairs of triangles. 
    \end{proposition}

    \begin{proof}
      \change{Assume that
      $T_v =\{u,v,w\}$ and 
      $T_{v'} = \{u',v',w'\}$ are weakly linked.}
      They can be weakly
      linked for two reasons. The first is that $T_v$ and $T_{v'}$ are
      compatible in $P[U]$, but at least one edge was deleted going
      from $P$ to $H$. The second case is that $T_v$ and $T_{v'}$ are
      also weakly linked in $P$. For the first type, since at most
      $2\eps \vol_H(U)$ edges were deleted in $P[U]$, this is also
      an upper bound on the number of weakly linked pairs of triangles of
      this type. Thus, the remainder of the proof is bounding the
      number of pairs of the second type.

      Let $\psi_G(u) = g_{u}$, and so
      forth. Since $T_v$ and $T_{v'}$ are weakly linked, we know that
      at least one (but not all) of
      $(g_u,g_{v'}), (g_u, g_{w'}), (g_v,g_{u'}), (g_v,g_{w'}), (g_w,
      g_{u'}), (g_w, g_{v'})$ is an edge of $G$. From this, we can
      deduce that either one of $g_u, g_v, g_w$ is contained in
      $\Delta(g_{u'},g_{v'},g_{w'})$ or one of
      $g_{u'}, g_{v'}, g_{w'}$ is contained in
      $\Delta(g_u, g_v, g_w)$. Further, a single
      $g \in \Delta(g_u,g_v, g_w)$ can \change{appear in} at most $3$ other
      triangles--namely those that the vertices
      $\psi^{-1}(a,g)$, $\psi^{-1}(b,g)$, $ \psi^{-1}(c,g)$ appear
      in. Therefore, we can upper bound the number of weakly linked
      triangles of this type by
      \[
        \sum_{\substack{v \in \mathcal B\\\{u,v,w\} = T_v}} 3|\Delta(g_u, g_v, g_w)|
      \]
      Using Proposition~\ref{prop:3-int}, we can upper bound this
      quantity by 
      \[
        \sum_{\substack{v \in \mathcal B\\\{u,v,w\} = T_v}} 50\eps
        (\Gamma_H(u) + \Gamma_H(v) + \Gamma_H(w)) = 50\eps \vol_H(U),
      \]
      where we use the fact that
      $\mathcal A \cup \mathcal B \cup \mathcal C$ is a partition of
      $U$. Therefore, there are at most $52\eps \vol_H(U)$ weakly
      linked pairs of triangles.
    \end{proof}
    
    Thus, by Proposition~\ref{prop:link}, we have that
    $|E(\widetilde{H}_U) \setminus E(H_U)| \le
    5\cdot 52\eps \vol_H(U)$. Therefore,
    $|E(\widetilde{H}_U) \Delta\\E(H_U)| \le 260\eps
    \vol_H(U) + |U|$ as desired.

    For the runtime, note that each phase of the algorithm can be done in
    $O(|Y_j|^3)$ time, including finding the max-min matching
    by performing a binary search on the bottleneck edge weight and
    then using a bipartite matching algorithm~\cite{CLRS}. 
  \end{proof}

\subsection{Proof of Theorem~\ref{thm:main}}

\begin{algorithm}[ht]
  \caption{Main Algorithm}\label{alg:main}
  \begin{algorithmic}
    \small
    \State \textbf{Input:} The graph $H$
    \State \textbf{Output:} Graph $\widetilde{H} = K_3 \times
\widetilde{G}$ on $V(H)$ with $ |E(H) \triangle E(\widetilde{H})| \leq  O(\epsilon |E(H)| )$, and $\widetilde{G}$
\State
    \State Construct graphs $C$, $\mathcal T(C)$. Compute connected
    components of $\mathcal{T}(C)$, $Y_1,\ldots, Y_\ell$
    \State Set $S = V(H)$ and $\mathcal F = \{\}$
    \For{each $Y_j$ in $Y_1, \hdots, Y_\ell$}
    \State Let $U_j$ be vertices induced by $Y_j$. 
    \If{two triangles of $Y_j$ share a common vertex} go to $Y_{j+1}$
    in for loop \EndIf 
    \State Run Algorithm \ref{alg:core-factoring} on $Y_j$ with $U =
    U_j \cap S$.
    \If{output is \textsf{FAIL}} go to $Y_{j+1}$ in for loop \EndIf
    \If{$|E(H) \cap (U \times (S \setminus U))| \ge
      5\eps/(1-\eps)\min(\vol_H(U), \vol_H(S \setminus U))$} go to
    $Y_{j+1}$ in for loop \EndIf
    \State Add factorization of $H[U]$ to $\mathcal F$.
    \State Set $S \leftarrow S \setminus U$.
    \EndFor
    \Return disjoint union of $\mathcal F$.
  \end{algorithmic}
\end{algorithm} 

We are now ready to prove Theorems \ref{thm:main}.

\begin{proof}[Proof of Theorem \ref{thm:main}]
  We set $S = V(H)$ and loop over the components $Y_j$ \change{of}
  $\mathcal T(C)$. We shall maintain the invariant that $S$ is
  atomic. For each component, we run
  Algorithm~\ref{alg:core-factoring} with $U = U_j \cap S$ (assuming
  $U \neq \emptyset$\change{)}, which by
  Lemma~\ref{lem:quasi-core-match} must either output \textsf{FAIL} or a
  factoring of $K_3 \times \widetilde{G}_U$ of $H[U]$ with error
  bounded by $260\eps \vol_H(U) + |U|$. If \textsf{FAIL} is output, we
  know that $Y_j$ cannot be core, so we can continue.

  By Proposition~\ref{prop:core-cut}, we know that if $Y_j$ is core
  then
  \change{$|E(H) \cap (U \times (S \setminus U))| \le\frac{5\eps}{1-\eps}
  \min(\vol_H(U),\vol_H(S \setminus U))$}. Thus, we can throw out
  any $Y_j$ which does not have this property.

  Since each vertex of $S$ is covered by a core component, we shall
  accept at least one $Y_j$, we can replace
  $S$ with $S \setminus U$ and continue in the loop. Note that by the
  time we loop over all the $Y_j$'s, we must have $S = \emptyset$,
  because each vertex of $S$ is included in some core
  component. Further, at each step $U_j$ is atomic so $U_j \cap S$ is
  atomic and $S \setminus U_j$ is also atomic. Thus, the invariants
  are maintained throughout the algorithm.

  Our final factorization is the disjoint union of all the
  factorizations found throughout the algorithm. Let $\widetilde{H}$
  be this graph, and let $U^{(1)}, \hdots,
  U^{(\ell)}$ be the vertices which induce the components of
  $\widetilde{H}$. The total error is then \change{(by Proposition \ref{prop:core-cut} and Lemma~\ref{lem:quasi-core-match})}
  \begin{align*}
    |E(\widetilde{H}) \Delta E(H)| &\le \sum_{i = 1}^{\ell}
                                     \left(|E(\widetilde{H}[U^{(i)}]) \Delta
                                     E(H[U^{(i)}])| +
                                     |E(H) \cap (U^{(i)} \times
                                     (U^{(i+1))} \cup \cdots \cup U^{(\ell)})|\right)\\
                                   &\le \sum_{i=1}^{\ell} \left(260 \eps \vol_H(U^{(i)}) +
                                     |U^{(i)}| +
                                     \frac{5\eps}{1-\eps}\vol_H(U^{(i)})\right)\\
                                    &\change{\le \sum_{i=1}^{\ell} \left(\left (520 \eps +\frac{10\eps}{1-\eps}\right )|E(H[U^{(i)}])|+
                                    |U^{(i)}| \right) =  \left (520 \eps +\frac{10\eps}{1-\eps}\right ) |E(H)|+
                                    |V(H)| }  \\
                                   &\le 550 \eps |E(H)|,
  \end{align*} 
  in the last line, we use that the sum of the degrees of the vertices of
  $H$ is twice the number of edges and the fact that the $|V(H)|$ term
  can be absorbed into the constant \change{for $\eps =\Omega( |V(H)|/|E(H)|)$.}

  For the runtime, computing the connected components of
  $\mathcal T(C)$ takes $O(n^6)$ time. Consider iteration $j$ of the
  loop. Note that checking if two triangles of $Y_j$ overlap takes
  $O(n)$ time at most. If $|Y_j| > n$, then the loop iteration takes
  $O(|Y_j|) = O(n)$ time, as the unique coverage of the vertices of
  $U$ must fail. Otherwise, if $|Y_j| \le n$, the loop iteration takes
  $O(|Y_j|^3 + n^2) = O(n^3)$ time as each $|Y_j| = O(n)$.  We know
  that $\ell = O(n^3)$, so the runtime is $O(n^6)$. 
\end{proof}

We remark that the runtime of our algorithm does not depend on $\eps$.
If one does not know $\eps$ before running the algorithm and is only given a graph $H$, 
one can binary search on its possible values to find a value for $\eps$ 
close to the smallest one for which $H$ is $\eps$-near some triangle tensor.
This only increases the runtime by a factor of $\log(1/\eps)$.

\subsection{Extensions}\label{sec:ext}

There are a \change{number of settings in which we believe Theorem~\ref{thm:main} and
Theorem~\ref{thm:no-sparse-cuts} can be extended. 
One such setting is allowing $G$ to be a small-set expander, and we discuss this below in detail.
We propose a few other settings, which require more formal study in order to obtain a reconstruction goal, in Section~\ref{sec:conclusion}.}

\paragraph{$G$ is a small-set expander} Recall that for
Theorem~\ref{thm:no-sparse-cuts}, we prove that $H$ can be efficiently
$3$-colored, if $G$ is a $3\eps$-edge-expander by showing that there
can be at most one core triangle component, which we can efficiently
3-color by Lemma~\ref{lem:quasi-core-match}. If instead a graph
$H$ has the property that there are at most $k$ core components, we
can in $n^{O(k)}$ time brute force guess which $k$ are the core ones,
determine their colorings and combine them to color $H$\footnote{We leave open the question of finding (or ruling out) \change{an FPT algorithm.
 That is, one with a runtime of
  $n^{O(1)} + f(k)$, where $f$ does not depend on $n$.}}.

We can prove such a property for $H$ with a broader class of graphs
$G$, which are known as \emph{small set expanders}. We say that $G$ is
a $(\delta, \alpha)$-small set expander if for every set $S$ of size at
most $\delta \cdot n$, we have that
\[
  |E(G) \cap S \times (V(G) \setminus S)| \ge \alpha \vol_G(S).
\] 

Using methods similar to the proof of
Theorem~\ref{thm:no-sparse-cuts}, we can show that if $G$ is a $(1/k,
\Omega(\eps))$-small set expander, then $H$ has at most $k$ core
components, and thus can be 3-colored in $n^{O(k)}$ time.

One interesting example here is the $\beta$-noisy hypercube, the
graph on $\{0,1\}^\ell$, for $n = 2^\ell,$ where nodes are adjacent if and
only if their hamming distance is $\beta \ell$.  This graph is a
$(1/\log(n), 1/2)$-small set expander (see for instance, \cite{zhao2021generalizations}), and we can obtain a 3-coloring in
quasi-polynomial time with our algorithm.

\section{Hardness}\label{sec: hardness}
In this section, we will show that certain natural tensor
reconstruction criteria are in fact NP-hard to achieve and prove
Theorem~\ref{thm:no-3-col}. As an immediate corollary of
Theorem~\ref{thm:no-3-col}, this rules out a \emph{monotone
  reconstruction} for tensor graphs.

\begin{corollary}
  Given as input an $H$ which is $\eps$-near $K_3 \times G$
  for some $G$, it is \textsf{NP}-hard to find a graph
  $\widetilde{H} = K_3 \times \widetilde{G}$ on the same vertex set as $H$ with
  $\widetilde{H} \supseteq H$.
\end{corollary}

\begin{proof}
  Assume we found such an $\widetilde{H}$. Then by
  applying Imrich's algorithm, we can factor a copy of $K_3$ out of
  $\widetilde{H}$ and thus determine a 3-coloring of
  $\widetilde{H}$. Since $\widetilde{H} \supseteq H$, this is also a
  3-coloring of $H$. Thus, by Theorem~\ref{thm:no-3-col}, finding
  $\widetilde{H}$ is \textsf{NP}-hard.
\end{proof}

\subsection{\change{3-Coloring with Equality}}

\change{We show that Theorem~\ref{thm:no-3-col} follows from the hardness of a problem we call ``3-coloring with equality.''}\footnote{\change{In the hardness of approximation community, this problem can be viewed as a special case of a problem known as \emph{Label Cover}}.} \change{An instance of 3-coloring with equality consists of a set of vertices with two edge sets: $(V, E_{\neq}, E_{=})$. We call the edges $E_{\neq}$ \emph{coloring constraints} and the edges $E_{=}$ \emph{equality constraints}. An \emph{assignment} is a map $\psi : V \to [3]$. We say this assignment \emph{satisfies} the instance if for all $\{u,v\} \in E_{\neq}$, we have that $\psi(u) \neq \psi(v)$ and for all $\{u,v\} \in E_{=}$ we have that $\psi(u) = \psi(v)$.}

\change{We say that a vertex $v \in V$ is \emph{$\eps$-loose} if} at most an $\eps$ fraction of its neighboring constraints are coloring constraints. \change{We say an instance of 3-coloring with equality is $\eps$-loose if every vertex is $\eps$-loose.} We now show that \change{3-coloring with equality} is
NP-hard even when \change{the instance is} $\eps$-loose for small $\eps$.

\begin{lemma}\label{lem:eps-coloring}
  Assume that $\eps \in (0,1)$ is constant. Given an instance
  $(V,E_{\neq},E_{=})$ of \change{3-coloring with equality} which is
  $\eps$-loose, it is NP-hard to determine if the instance is
  satisfiable.
\end{lemma}

\begin{proof}
  We reduce from the NP-hardness of 3-coloring. \change{Let $G := (V_3, E_3)$ be an instance of $3$-coloring. Let $d_v$ be the degree of each $v \in V_3$. We now define our instance $H := (V, E_{\neq}, E_{=})$ of 3-coloring with equality.}

  For each $v \in \change{V_3}$, create $n_v := \lceil d_v/\eps\rceil$ new
  vertices $v_1, \hdots, v_{n_v}$, which we call $v$'s \emph{cloud
    vertices}. \change{Let $C_v$ be this set of cloud vertices.}

  \change{We then define
  \begin{align*}
    V &= V_3 \cup \bigcup_{v \in V_3} C_v\\
    E_{\neq} &= E_3\\
    E_{=} &= \bigcup_{v \in V_3} \bigcup_{v' \in C_v} \{v, v'\}.
  \end{align*}}
  First, note that $H\change{= (V, E_{\neq}, E_{=})}$ is $\eps$-loose, as each vertex
  corresponding to a vertex of $G$ is $\eps$-loose. Further, each
  of the cloud vertices has a single equality constraint and thus is
  $\eps$-loose.

  Now, if $G$ is 3-colorable, then $H$ can be satisfied by copying
  the coloring for the copies of the vertices in $G$, and then
  assigning the cloud vertices to have the same color as the vertex
  they are adjacent to. Likewise, if $H$ is satisfied, the cloud
  vertices can be deleted to yield a 3-coloring of $G$. Thus, we
  have completed the reduction.
\end{proof}

\subsection{\change{The Tensor Reduction}}

\change{We} give an efficient
reduction from an instance of \change{3-coloring with equality} to a tensor
reconstruction problem (such as in \cite{DMR09,BG16}). \change{Let $(V, E_{\neq}, E_{=})$ be our instance of 3-coloring with equality, and let $(V', E')$ be the graph for our tensor reconstruction problem.}

\change{Our vertex set is $V' = V \times [3]^3$, that is there are $27$ copies of each vertex.  For each $v \in V$ and $x, y \in [3]^3$, we have that $\{(v, x), (v, y)\} \in E'$ if and only if $x_i \neq y_i$ for all $i \in [3]$. In other words, the subgraph induced by $\{v\} \times [3]^{3}$ is isomorphic to the tensor product of $3$ copies of $K_3$.}

\change{For each $\{u, v\} \in E_{\neq}$ and $x, y \in [3]^3$, we have an edge between $\{(u,x), (v,y)\} \in E'\}$ if and only if $x_i \neq y_j$ for all $(i, j) \in K_3 := ([3], \neq)$. For each $\{u,v\} \in E_{=}$ and $x, y \in [3]^3$, we have $\{(u,x), (v,y)\} \in E'$ if and only if $x_i \neq y_i$ for all $i \in [3]$. }

Theorem~\ref{thm:no-3-col} follows from the following two observations (which
are proved in the following to subsections).

\begin{lemma}[Completeness]\label{lem:completeness}
  \change{For $c_{\textsf{loose}} < 1/3$, given} a $c_{\textsf{loose}}\eps$-loose instance \change{$(V,E_{\neq}, E_{=})$ of $3$-coloring with equality} which is satisfiable, the \change{tensor} reduction produces a graph $H$ which is $\eps$-near $K_3 \times G$ for some $G$.
\end{lemma}

\begin{lemma}[Soundness]\label{lem:soundness}
  Given an instance \change{$(V,E_{\neq},E_{=})$ of $3$-coloring
    with equality} with no satisfying assignment, the \change{tensor} reduction
    produces a graph $\change{H}$ which is not 3-colorable.
\end{lemma}

\begin{proof}[Proof of Theorem~\ref{thm:no-3-col}]
  Let $\mathcal A$ be an algorithm for $3$-coloring graphs $\eps$-near
  $K_3 \times G$ for some $G$. By Lemma ~\ref{lem:eps-coloring}, it
  suffices to show that a polynomial-time modification of
  $\mathcal A$ can solve $c_{\textsf{loose}}\eps$-loose instances of \change{3-coloring with equality.}

  Our modified algorithm $\mathcal A'$ is to take an instance \change{$(V,E_{\neq},E_{=})$ of $3$-coloring
    with equality}, apply the \change{tensor}
  reduction (which is polynomial-time) to produce a graph
  $\widetilde{H}$ and then apply $\mathcal A$ to $\widetilde{H}$. We then
  check that the output of $\mathcal A$ is indeed a 3-coloring of
  $\widetilde{H}$. If so, we output ``satisfiable.'' Otherwise, we output
  ``unsatisfiable.''

  If the \change{3-coloring with equality} instance is satisfiable, by
  Lemma~\ref{lem:completeness}, we have that $\widetilde{H}$ is
  $\eps$-near a tensor with $K_3$. Thus, $\mathcal A$ will output a
  valid 3-coloring, and thus we shall output ``satisfiable'' for
  this \change{3-coloring with equality} instance.

  On the other hand, if the \change{3-coloring with equality} instance is not satisfiable,
  then $\widetilde{H}$ is not even 3-colorable by Lemma
  \ref{lem:soundness}.  Thus, even if $\mathcal A$ outputs a coloring,
  it is not a valid 3-coloring, and thus we shall output
  ``unsatisfiable.''

  Thus, finding a 3-coloring of a graph $\eps$-near $K_3 \times G$
  is NP-hard.
\end{proof}

\subsection{Completeness: Proof of Lemma~\ref{lem:completeness}}

Let $\psi : V \to [3]$ be a satisfying assignment to our $c_{\textsf{loose}}\eps$-loose \change{instance $(V,E_{\neq},E_{=})$ of $3$-coloring with equality. Let $H$ be the graph produced by the tensor reduction.} We show how to \change{approximately} factor this graph as a $K_3$ times another graph \change{$G$. We first build the graph $G$ and then construct the approximate isomorphism between $K_3 \times G$ and $H$.}

\change{The vertices of $G$ are $V \times [3]^2$. For any $v \in V$ and $x, y \in [3]^2$, we have that $\{(v,x), (v,y)\} \in E(G)$ if and only if $x_1 \neq y_1$ and $x_2 \neq y_2$. Likewise, for any $\{u, v\} \in E_{=}$ and $x, y \in [3]^2$, we have that $\{(u,x), (v,y)\} \in E(G)$ if and only if $x_1 \neq y_1$ and $x_2 \neq y_2$. For each $\{u, v\} \in E_{\neq}$ and $x, y \in [3]^2$, we always have that $\{(u,x), (v,y)\} \in E(G)$.}

\change{We now construct a map $\pi : [3] \times V(G) \to V(H)$. For each $(v, x) \in V(G)$ and $y \in [3]$, we map the pair $(y, (v, x))$ based on the value of $\psi(v)$:
\[
\pi(y, (v,x))= \begin{cases}
                   (v, (y, x_1, x_2)) & \psi(v) = 1\\
                   (v, (x_1, y, x_2)) & \psi(v) = 2\\
                   (v, (x_1, x_2, y)) & \psi(v) = 3.
                 \end{cases}
\]}

\change{We now show that $E(H) \subseteq E(\pi(K_3 \times G))$, and that the lack of equality is only due to edges added to $G$ via edges of $E_{\neq}$.}

\change{Observe that for any edge $\{y, y'\} \in K_3$ and $\{(v, x), (v, x')\} \in V(G)$, we have that the edge arising from the tensor product $\{\pi(y, (v,x)), \pi(y', (v,x'))\}$ is an edge of $H$ as $y \neq y', x_1 \neq x'_1$ and $x_2 \neq x'_2$. Further, this is a bijection between the subgraph of $H$ induced by $\{v\} \times [3]^3$ and the tensor product of $K_3$ with the subgraph of $G$ induced by $\{v\} \times [3]^2$.

Likewise, if $\{u, v\} \in E_{=}$, for any $\{y, y'\} \in K_3$ and $\{(u, x), (v, x')\} \in V(G)$, we have that the corresponding product $\{\pi(y, (u,x)), \pi(y', (v,x'))\}$ is an edge of $H$. Further, the edges of $H$ arising from $\{u, v\} \in E_{=}$ are in bijection with the tensor product of $K_3$ with the subgraph of $G$ with the edges added due to $\{u, v\}$.

If $\{u, v\} \in E_{\neq}$, we no longer have such a bijection between $K_3$ times the edges added to $G$ and the edges added to $H$. However, we can prove that every such edge of $H$ corresponds to an edge in $K_3 \times G$. Assume without loss of generality that $\psi(u) = 1$ and $\psi(v) = 2$. Consider $x, y \in [3]^3$ with $x_i \neq y_j$ for all $(i,j) \in K_3$. Then, note that
\begin{align*}
  \pi^{-1}(u, x) &= (x_1, (u, (x_2, x_3)))\\
  \pi^{-1}(v, y) &= (y_2, (v, (y_1, y_3)))\\
\end{align*}
As stated $x_1 \neq y_2$, so $\{x_1, y_2\} \in K_3$. Further, $\{u, v\} \in E_{\neq}$ so $\{(u, (x_2, x_3)), (v, (y_1, y_3))\} \in E(G)$.
}

\change{Therefore, $E(H) \subseteq E(\pi(K_3 \times G))$, and that the lack of equality is only due to edges added to $G$ via edges of $E_{\neq}$.}

\change{To finish, for every $(u,x) \in V(H)$, we need to bound the number of incident edges in $E(\pi(K_3 \times G)) \setminus E(H)$. For each $\{u, v\} \in E_{\neq}$, there are $18$ edges between $(u,x)$ and vertices of the form $(v,y) \in \pi(K_3 \times G)$. Further, for each $(u,v) \in E_{=}$, there are $8$ edges between $(u,x)$ and vertices of the form $(v,y) \in V(H)$. Finally, there are $8$ edges between $(u,x)$ and vertices of the form $(u,y) \in V(H)$.

Let $a_{u} = |\{v \in V : \{u,v\} \in E_{=}\}|$ and $b_{u} = |\{v \in V : \{u,v\} \in E_{\neq}\}|$. Since our 3-coloring with equality instance is $c_{\textsf{loose}}\eps$-loose, we have that $b_u \le c_{\textsf{loose}}\eps(a_u + b_u)$.  Therefore, the fraction of edges of $\pi(K_3 \times G))$ incident to $(u,x)$ which are in $E(\pi(K_3 \times G)) \setminus E(H)$ is bounded by
\[
  \frac{18b_u}{8 + 8a_u + 18b_u} \le \frac{18b_u}{8a_u + 8b_u} = \frac{18}{8} c_{\textsf{loose}}\eps < \eps,
\]
as $c_{\textsf{loose}} < 1/3$.

Therefore, the graph $H$ produced by the tensor reduction is $\eps$-near $K_3 \times G$, as desired.
}
\subsection{Soundness: Proof of Lemma~\ref{lem:soundness}}

We prove the soundness via the contrapositive. Assume there exists a
3-coloring of the graph formed via the \change{tensor} reduction. That is, \change{there exists a map $\psi : V \times [3]^3 \to [3]$ forming a 3-coloring.} A key result, taken from the literature, is that \change{for each $v \in V$, the induced map $\psi(v, -) : [3]^3 \to [3]$} can be unambiguously decoded \change{to a single color $\phi(v) \in [3]$.}

\begin{claim}[\cite{Greenwell1974}]\label{claim:unique}
  Given a graph $T$ which is the tensor of $L \change{\ge 1}$ copies of $K_3$ and a
  3-coloring $c : [3]^L \to [3]$ of $T$, there exists a
  $i \in [L]$ and a permutation $\eta : [3] \to [3]$ such that for all
  $x \in [3]^L$, we have that $c(x) = \eta(x_i)$.
\end{claim}

\begin{proof}[Proof of Lemma~\ref{lem:soundness}]
For each $v \in V$, let $i_v$ and $\eta_v$ be such that for all \change{$x \in [3]^3$, $\psi(v, x) = \eta_v(x_{i_v})$. For all $v \in V$, define $\phi(v) = i_v$. We claim that $\phi$ is a satisfying assignment to the instance $(V, E_{\neq}, E_{=})$ of 3-coloring with equality.

First, consider any equality constraint $\{u,v\} \in E_{=}$. We seek to show that $\phi(u) = \phi(v)$. Assume for sake of contradiction that $\phi(u) \neq \phi(v)$. Without loss of generality, we can assume that $\phi(u) = 1$ and $\phi(v) = 2$. Then, note that for any $x, y \in [3]^3$, we have that $\psi(u,x) = \eta_u(x_{\phi(u)}) = \eta_u(x_1)$ and $\psi(v,y) = \eta_v(y_2)$. Note that for any choice of $\bar{x}_1, \bar{y}_2 \in [3]$, there exists $x, y \in [3]^3$ with $x_1 = \bar{x}_1$, $y_2 = \bar{y}_2$ and $x_i \neq y_i$ for $i \in [3]$. Thus, $\eta_u(\bar{x}_1) \neq \eta_v(\bar{y}_2)$ for all $\bar{x}_1, \bar{y}_2 \in [3]$, a contradiction. Therefore, $\phi(u) = \phi(v)$.

Second, consider any coloring constraint $\{u,v\} \in E_{\neq}$. We seek to show that $\phi(u) \neq \phi(v)$. Assume for sake of contradiction that $\phi(u) = \phi(v)$. Without loss of generality, we can assume that $\phi(u) = \phi(v) = 1$. Thus, for any $x, y \in [3]^3$, we have that $\psi(u,x) = \eta_u(x_{\phi(u)}) = \eta_u(x_1)$ and $\psi(v,y) = \eta_v(y_1)$. 

Note that for any choice of $\bar{x}_1, \bar{y}_1 \in [3]$, there exists $x, y \in [3]^3$ with $x_1 = \bar{x}_1$, $y_1 = \bar{y}_1$ and $x_i \neq y_j$ for $i, j \in [3]$ with $i \neq j$. If $\bar{x}_1 \neq \bar{y}_1$, set $x = (\bar{x}_1, \bar{x}_1, \bar{x}_1)$ and $y = (\bar{y}_1, \bar{y}_1, \bar{y}_1)$. If $\bar{x}_1 = \bar{y}_1$, set $x = (\bar{x}_1, \bar{x}_1+1, \bar{x}_1 + 1)$ and $y = (\bar{x}_1, \bar{x}_1+2, \bar{x}_1 + 2)$.

Thus, $\eta_u(\bar{x}_1) \neq \eta_v(\bar{y}_1)$ for all $\bar{x}_1, \bar{y}_1 \in [3]$, a contradiction. Therefore, $\phi(u) \neq \phi(v)$. Therefore $\phi$ is a satisfying assignment of the 3-coloring with equality instance.}
\end{proof}

This completes the proof of Theorem~\ref{thm:no-3-col}. \change{We now prove the following corollary.

\begin{corollary}
  For any $\delta \in (0,1)$, given as input a graph $\widetilde{H}$
  on $N$ vertices which is $\eps = N^{-\delta} $-near $K_3 \times G$
  for some $G$, it is \textsf{NP}-hard to find a $3$-coloring of
  $\widetilde{H}$.
\end{corollary}

\begin{proof}
  We perform the same reductions as Theorem~\ref{thm:no-3-col} from 3-coloring to 3-coloring with equality to tensor reconstruction. However, we need to verify that each step of the reduction accommodates non-constant $\eps$. We start with an instance $(V, E)$ of 3-coloring on $n$ vertices such that every vertex has degree $4$. This is NP-hard by a result of Dailey~\cite{dailey1980uniqueness}. We now perform the reduction of Lemma~\ref{lem:eps-coloring}, but $\eps = n^{-\eta}$ for some $\eta > 0$ to be specified later. In particular, each vertex $v \in V$ becomes a cloud of $\lceil 4n^{\eta}\rceil$ vertices. Thus, the instance $(V', E'_{\neq}, E'_{=})$ of 3-coloring with equality has $|V'| \le n(1 + \lceil 4n^{\eta}\rceil) \le 6n^{1+\eta}$ vertices and is $n^{-\eta}$-loose. The tensor reduction then produces an instance $\widetilde{H}$ of tensor reconstruction on $N := 27(6n^{1+\eta})$ vertices. By Lemma~\ref{lem:completeness} and Lemma~\ref{lem:soundness}, it is NP-hard to distinguish whether $\widetilde{H}$ is $\eps'=\frac{n^{-\eta}}{3}$-near a $K_3$-tensor or is not 3-colorable. In terms of $N$, we have that $\eps' = O(N^{-\eta/(1+\eta)})$. If we pick $\eta$ sufficiently large such that $\frac{\eta}{1+\eta}$ is strictly greater than $\delta \in (0, 1)$, then $\eps' < N^{-\delta}$, as desired.
\end{proof}
}

\section{Conclusions and Open Questions}\label{sec:conclusion}

Inspired by finding novel tractable instances of the 3-coloring
problem, in this paper we have studied the \change{efficient} 3-colorability
of graphs that are approximately of the form $K_3 \times G$. In
particular, if $G$ is a mild expander, then 3-coloring is indeed
possible. However, it is NP-hard for general $G$, although weaker
tensor reconstruction criteria hold, \change{such as the $\ell_1$ reconstruction goal achieved in Theorem \ref{thm:main}}. 
There are a number of directions
which could be pursued to extend the results in this paper. We list a
few of the most promising ones.

\change{
First, we believe one can extend our results to more general settings.
We discuss a few such settings that seem like very tractable future work, 
as we think here, one  can make use of our techniques.}

\paragraph{Allowing deletions and insertions} One can change the
model from strictly deletions to allowing deletions and insertions, so
long as the insertions respect the tensor structure.  In other words,
the adversary can insert an edge between $u$ and $v$ only if
$\psi_{K_3}(u) \neq \psi_{K_3}(v)$ and $\psi_{G}(u) \neq \psi_{G}(v)$. We observe that this model can be
reduced to the deletion model that we have studied.\footnote{A notable
  analogy in coding theory is that any error-correcting code for a
  deletion channel is also an error-correcting code for the
  insertion-deletion channel (c.f. \cite{levenshtein1966binary}).} Let
$G' \supset G$ be the graph where $g$ and $g'$ are \change{connected if there is
an edge added } between and $u$ and $v$ with $\psi_G(u) = g$ and
$\change{\psi_G(v) =g'}$. Observe that $H$, the graph the adversary produced, is a
subset of $P' := K_3 \times G'$. Further, for any vertex $u$ of $H$,
the number of incident edges in $P' \setminus H$ is at most
$6\eps |\Gamma_{P'}(u)| \le \frac{6\eps}{1-6\eps} |\Gamma_{P}(u)|$.
\change{Thus, we belive one can analyze $H$ as the result of deleting
$O(\eps)$ fraction of the edges from each vertex of $P'$, and 
perhaps obtain similar guarantees to those in Theorems \ref{thm:main} and
\ref{thm:no-sparse-cuts}.}

Further, one can allow insertions of edges $(u,v)$ for which
$\psi_{G}(u) = \psi_{G}(v)$.  At most two such edges can be added per
vertex and so $\eps$ can be increased very slightly in each of the
inequalities to accommodate these edges. Thus, we conjecture that the most general
insertion model our algorithm can withstand is one in which the
original color classes are respected.

\paragraph{Extending past $K_3$}
We \change{believe that one can use analogous arguments to the ones presented in this section
to replace $K_3$ with any fixed size \emph{core graph}, which is a graph such that every homomorphism on it is an isomorphism~\cite{hell1992core}.
Some examples of core graphs are cliques and an odd-length cycles.}
Call this fixed graph $F$, and let $f$ denote the number of
vertices in $F$.  \change{In these settings, there are an analogue of core
triangles.  First, one can build an analog of the graph $C$ that connects
nodes that are candidates to be from the same $F$ component--that is,
the sizes of the pairwise neighborhoods are proportional to the sizes
of the pairwise neighborhoods of $F$. } Then, one could form the equivalent of
the graph $\mathcal{T}(C)$ by considering tuples of size $f$ that form
copies of $F$ in $C$, and such tuples are compatible with each other
if they have the edges between them in $H$ that we would expect from
the tensor. We conjecture that one can then prove that these components are either
``core'' or ``monochrome.,'' though it is likely that the constants and runtime will get worse as $f$
increases.

\change{\begin{question}
Let $F$ be a fixed size clique or an odd-length cycle and 
let $G$ be a $O(\eps)$-expander. Let $H$ be $\eps$-near $F
\times G$. 
Given $H$, can one efficiently find a tensor graph $\widetilde{H}$
which is $\eps$-near $F \times G$?
\end{question}}

\paragraph{A spectral variant of Imrich's algorithm} The algorithms
we present in this paper are highly combinatorial. Yet, as evidenced
by Theorem~\ref{thm:no-sparse-cuts}, the spectral properties of $G$ are essential for
learning some structure of the \change{underlying} graph. Thus, a natural goal would be to
find a truly spectral variant of Imrich's algorithm (e.g., via semidefinite
programming) which allows for robust reconstruction. In particular,
the following result could potentially result from such an investigation.

\begin{question}
  Let $F$ and $G$ be $O(\eps)$-expanders. Let $H$ be $\eps$-near $F
  \times G$. Given $H$, can one efficiently find a tensor graph $\widetilde{H}$
  which is $\eps$-near $F \times G$?
\end{question}

Perhaps related to this question is the work of
Trevisan~\cite{trevisan2012max} which relates high-quality
approximations of MAX CUT to spectral properties of the input
graph.

\paragraph{Alternative reconstruction goals} 
We show that when $H$ is $\eps$-near $K_3 \times G$,
we can find a tensor $\widetilde{H} = K_3 \times \widetilde{G}$ 
with $|E(\widetilde{H}) \Delta E(H)|\leq O(\eps |E(H)|)$.
Since deletions are bounded on each vertex, 
another seemingly reasonable reconstruction goal is to try to find $\widetilde{H}$
that is close to $H$ on every \change{vertex}, e.g. 
for all $v \in H$, $|\Gamma_H(v) \Delta \Gamma_{\widetilde{H}}(v)|\leq O(\eps |\Gamma_H(v)|)$. 
This is just one other example reconstruction goal, and 
there may be other natural, tractable reconstruction goals for robust tensor factoring
that are worth defining and studying.

\paragraph{Extensions to other deletion models} The deletion model we
assumed in this paper is that $\eps$ fraction of the edges incident to
each vertex were deleted. Another natural model is that $\eps$
fraction of the total edges are deleted. Analyzing this model appears
more difficult as many vertices could have a large fraction of their
neighbors deleted, which cause our heuristics to fail.

Another natural model to consider is random deletions with
i.i.d. deletion probability $\eps$. With high probability, the results
of the main theorems apply, but it may be possible to 3-color
random deletions of $K_3 \times G$ for an adversarial
$G$. ``Semi-random'' graph coloring problems such as this have been of
interest in the literature (c.f., \cite{feige1998heuristics}).

\paragraph{Interesting subclasses of 3-colorable graphs}
We might be able to better understand hardness of the 3-coloring problem 
by studying interesting special cases.
Indeed, this was the motivation behind our work, 
as well as that of Blum in coloriong random 3-colorable graphs \cite{Blum94},
Arora and Ge \cite{arora2011new} in studying 3-colorable graphs with low threshold rank,
and Kawarabayashi and Thorup \cite{kawarabayashi2017coloring} in focusing on graphs with high degree.
Based on these previous works, a natural extension of these ideas is the following question:
\begin{question}
  What other structural properties of a 3-colorable graph imply one can color it 
with few (at most $O(\log n)$) colors efficiently?
\end{question}

\paragraph{Deeper connections to coloring} This study was inspired by
trying to better understand how to color 3-colorable graphs with a
small number of colors. An important direction is to better understand
the implications our results have for the approximate coloring
problem. For instance, if a graph can be efficiently partitioned into
approximate tensors, can we use our reconstruction algorithm on the
approximate tensor pieces and some other meta algorithm on the full
graph.  More generally, we ask the following:
\begin{question}
  Can we bootstrap our algorithm to give
approximate colorings of general graphs? 
\end{question}

\change{\paragraph{Stronger hardness results} One may seek to
  strengthen Theorem~\ref{thm:no-3-col} by imposing a weaker
  reconstruction criteria than finding a 3-coloring of
  $\widetilde{H}$. By adapting the methods of \cite{BG16}, it should
  be possible to show that it is NP-hard to find even a 4-coloring of
  $\widetilde{H}$. However, the approach we take is useless even for
  six colors, per observations of \cite{BG16,BBKO21}, unless one makes
  conditional assumptions like in \cite{DMR09}.

  One may also seek to instead show that it is NP-hard to reconstruct
  $\widetilde{H}$ with respect to some reconstruction metric. We not
  currently aware of a method to obtain such results.}

\section*{Acknowledgments}

This project wouldn't have been possible without the support of Tselil
Schramm, particularly for helping us shape the right questions and
models for our investigation.

We thank Janardhan Kulkarni for valuable discussions on coloring during the early
stages of the project. We thank Venkatesan Guruswami for useful feedback on the
project and Kasper Lindberg for helpful comments on an early draft of the paper.
Finally, we thank Microsoft Research, Redmond, for virtually
hosting us, leading to the creation of this project.

\change{We thank anonymous reviewers for their careful reading and
  numerous helpful comments, including correcting an error in
  Section~\ref{sec: hardness}.}

\clearpage

\bibliographystyle{siamplain}
\bibliography{ref}

\end{document}